\documentclass[journal]{journal}

\usepackage{multicol,lipsum}
\ifCLASSINFOpdf
	\usepackage[pdftex]{graphicx}
 
\else
\fi
\usepackage[cmex10]{amsmath}

\usepackage{mathtools}

\usepackage{float}

\usepackage[export]{adjustbox}

\usepackage{fixltx2e}
\usepackage{ulem}

\usepackage{xcolor}

\usepackage{nameref}

\newtheorem{problem}{Problem}

\newtheorem{theorem}{Theorem}

\newtheorem{proof}{Proof}

\begin{document}
%
\title{A System Model-Based Approach for the Control of Power Park Modules for Grid Voltage and Frequency Services}
%
%
%

\author{{Bogdan MARINESCU,
        Elkhatib KAMAL,  
				Hoang-Trung NGO
        }
        
\thanks{Bogdan MARINESCU, Elkhatib KAMAL and Hoang-Trung NGO are with Ecole Centrale Nantes -LS2N (Laboratoire des sciences du numérique de Nantes), 1 Rue de la Noë, 44000 Nantes Cedex 3, France, Email: Bogdan.Marinescu@ec-nantes.fr, Elkhatib.Ibrahim@ec-nantes.fr and Hoang-Trung.Ngo@ec-nantes.fr.}
\thanks{This work is part of the H2020 European project POSYTYF (https://posytyf-h2020.eu/).}
}
%
%

\markboth{Journal of \LaTeX\ Class Files,~Vol.~6, No.~1, January~2007}%
{Shell \MakeLowercase{\textit{et al.}}: Bare Demo of IEEEtran.cls for Journals}
%



\maketitle
\thispagestyle{empty}

\begin{abstract}

A new control approach is proposed for the grid insertion of Power Park Modules (PPMs). It allows full participation of these modules to ancillary services. This means that, not only their control have some positive impact on the grid frequency and voltage dynamics, but they can effectively participate to existing primary and secondary control loops together with the classic thermal/inertia synchronous generators and fulfill the same specifications both from the control and contractual points of view. 

To achieve such level of performances, a \textit{system approach} based on an \textit{innovatory control model} is proposed. The latter control model drops classic hypothesis for separation of voltage and frequency dynamics used till now in order to gather these dynamics into a small size model. 

From the \textit{system point of view}, dynamics are grouped by time-scales of phenomena in the proposed control model. This results in more performant controls in comparison to classic approaches which orient controls to physical actuators (control of grid side converter and of generator side converter). Also, this allows coordination between control of converters and generator or, in case of multimachines specifications, among several PPMs. 

From the \textit{control synthesis point of view}, classic robust approaches are used (like, e.g., H-infinity synthesis).

Implementation and validation tests are presented for wind PPMs but the approach holds for any other type of PPM.

These results will be further used to control the units of the new concept of \textit{Dynamic Virtual Power Plant} introduced in the H2020 POSYTYF project.
\end{abstract}


%
\IEEEpeerreviewmaketitle

\section{Introduction}
%
%
%
%
\IEEEPARstart{R}{apid} development of Renewable Energy Sources (RES) brought the concept of Power Park Modules (PPMs) \cite{7x}. Indeed, RES, like, e.g., solar and wind \cite{1} are systematically connected to the grid via power electronics (like back-to-back power converter structures) (e.g., \cite{1}, \cite{2}).\\

PPMs and, generally, power electronics bring fast dynamics into power systems \cite{7}. Also, compared with dynamics of classic synchronous thermal generators, voltage and power dynamics are coupled in the same range of frequency. Old hypothesis and way of doing which separate control of voltage and power tracks should be revisited into more coordinated controls which need new concepts of modeling and regulation. 

Massive integration of RES lowers also the global inertia of the power system \cite{2}, \cite{7}. The Rate of Change of Frequency (RoCoF) is increased \cite{8x}, \cite{9x} which results in a need for fast frequency support from RES. Several approaches dealt with this as hidden inertia control \cite{7}, \cite{8}-\cite{12} or fast power reserve \cite{7}, \cite{13}-\cite{16}.

Intermittency and volatility of RES is a major difficulty for RES participation to ancillary services. This is obviously due to meteorological variations but also to technological difficulties of PPMs to stay connected to grid during large disturbing events. The latter point is crucial for increasing the actual RES penetration rate and it is a hot subject for Transmission System Operators (TSO), regulatory bodies and manufacturers \cite{entsoeGridCode}. Improvement of the controls is a key issue to overcome this.

Despite several recent research work, RES and PPMs are not fully integrated is actual secondary regulation schemes. In most cases, some grid frequency and voltage support is provided but in an indirect way. For example, the classic droop control and its variations \cite{7}, \cite{17}-\cite{18} and Maximum Power Point Tracking Algorithms (MPPT) \cite{2} \cite{3} do not allow a full integration of RES into secondary regulation and market contracts at the same level as the classic synchronous generators. To go into this direction, the control approaches are reviewed and improved at a \textit{system level} in this paper: first, a new control model is introduced to capture in a simple, direct and efficient way both frequency and voltage dynamics of all dynamic devices of interest (converters and generator of one or several PPMs). Next, control objectives and actions are distributed in a new optimal way (from the time scales point of view) opposed to allocation to each actuator (converter) as in the classic approaches. Finally, coordination and robustness of the resulting closed-loop system are improved by using advanced multivariable robust synthesis methods (like, e.g., H-infinity).

The paper is organized as follows: in Section II the objectives are formulated from the overall \textit{system point of view}. Classic hypothesis and approaches for control are critically and constructively revisited towards new approaches. In Section III, the control objectives are formulated from the \textit{control point of view} along with a new control framework. This new control is applied to wind PMSG in Section IV. The new control is based on a H-infinity regulator which is presented in Section V. Simulation results that illustrate the effectiveness of the proposed strategies are presented in Section VI. Conclusions and future prospects are presented in Section VII.

\section{System objectives and review of classic hypothesis and approaches}\label{section_SystemObjectives}
Control specifications of a grid connected generator are generally twofold: local (machine) specifications and global (grid) specifications. The local ones are to ensure good operation of the machine and manage its active (P) and reactive (Q) power generation. In function of technology, other variables might be regulated, as for example, blade pitch/position and shaft speed if a wind turbine is involved in the PPM. Global objectives target ancillary services. This means also to manage P and Q generation but with different specifications. Indeed, for this, the PPM should regulate its voltage not only for secure and optimal run of the generator itself but also for the neighbor AC grid. This means that the voltage in different (distant) strategic points of the grid should have specific response in case of several (prespecified by TSOs) dimensioning grid incidents. Also, the grid frequency should be maintained. In the measure of possible, large PPMs should also fully participate to secondary voltage (V) and frequency (f) regulations, with the same duties and rights as classic synchronous generators. This means also participation to market mechanisms. 

Traditionally, V and f regulations are carried out in an independent and separate manner. This is due to the fact that V and f dynamics are naturally separated in power systems which contain only large inertia synchronous generators (see, e.g., \cite{Kundur}). After their synthesis, the 2 controllers are installed on the generator and no parasitic interaction of their dynamics is registered because of the decoupling of the two dynamics. 

This work hypothesis is less and less valid in case of power systems with high rate of power electronics. Indeed, for a power converter, rapid control for all dynamics is possible and needed. V and f are no longer decoupled both for grid following and grid forming modes \cite{1x}, \cite{2x} of control. Coordination should thus be achieved in control.

Classic vector control is based on separation of the control actions according to the actuators. For the example of a wind PPM given in Fig. \ref{fig:PMSG}, a control is implemented for the grid side converter (GSC) and another one for the machine side converter (MSC) (see, e.g., \cite{vector_control_1}-\cite{vector_control_6}). 


Such separation of control loops is also inherited in many more advanced controls \cite{3x}, \cite{4x}.

Also, in vector control, dynamics are artificially separated by the two layers control structure: for each converter, poles of an inner loop control are placed at faster locations as the ones of the outer loop control (see Appendix \ref{subsectionVectorControl} where principles of such control are briefly recalled).  This time-scale separation allows one to synthesize controllers one by one (even though interactions still exist). Notice that this is not necessary as multivariable coordinated control can be performed as shown in the next section. It is nor optimal since changing the dynamics of the system will result in lowering robustness of the resulting closed-loop (see, e.g., \cite{Post}).

\section{Control objectives and a new hierarchical time-decoupling and fully coordinated control strategy}
\label{section:time_decoupling_and_fully_coordinated_control_strategy}

\subsection{Control objectives}
The facts mentioned above conducted us to a new control approach in which objectives and dynamics are managed closer to the open-loop dynamics of the system. Also, the system is considered as global as possible, i.e., the PPMs to be controlled and the power system in which they are inserted. For each PPM, the control objectives are translation of the physical/system objectives explained in Section \ref{section_SystemObjectives} at both local and grid levels:

\begin{itemize}
	\item local Q or (equivalently) V control, i.e., tracking a reference $Q_{ref}$ or $V_{ref}$ at the grid connection point. The latter are elements of vector $Y^{ref}$ in Fig. \ref{fig:Time_and_space_separation}.
	\item regulation of the DC bus voltage of the power electronics part, i.e., tracking of a reference $V_{DC_{ref}}$ (element of vector $Y_{0}^{ref}$ in Fig. \ref{fig:Time_and_space_separation})
	\item regulation of the generated active power to a given setpoint $P_{ref}$ (element of vector $Y^{ref}$ in Fig. \ref{fig:Time_and_space_separation})
	\item primary grid frequency support: minimization of the gap $\Delta f$ from actual grid frequency to nominal frequency (50 or 60Hz) in case of grid faults and variations.
	\item secondary grid frequency support: issue the reference $P_{ref}$ from the secondary frequency control level (AGC) of the AC zone to which the PPM belongs to (which contributes to $Y_{1}^{ref}$ in Fig. \ref{fig:Time_and_space_separation})
	\item secondary grid voltage support: issue the reference $V_{ref}$ for the voltage of the PPM grid connection point from the secondary voltage control level of the AC zone to which the PPM belongs to ($Y_{2}^{ref}$ in Fig. \ref{fig:Time_and_space_separation})
\end{itemize}

\subsection{Structure of the control}
To reach these objectives, a new control framework is proposed to reach maximum coordination between all PPMs actuators. The main difference with the classic vector control is that the control is not structured around each actuator, but according to the time response (frequency band) of the actuators and open-loop plant dynamics. This led us to the structure in Fig. \ref{fig:Time_decoupling}. Plant dynamics have been split into 3 categories:
\begin{itemize}
	\item very fast ones which correspond to dynamics of voltage/reactive power and active power variation needed to improve the grid frequency gradient (RoCoF).
	\item fast ones related to primary frequency regulation of the PPM. The classic control here is the \textit{droop control}.
	\item slow dynamics (more than one minute) which correspond to PPM participation to \textit{secondary grid controls}
\end{itemize}

Based on this time-scale separation, the proposed control method is structured as in Fig. \ref{fig:Time_and_space_separation}: Plant 0 is the plant to be controlled, i.e., the PPM and all grid dynamics of interest. It will be shown in the next section how the latter ones will be captured in an original \textit{control model}. Three stages of control are proposed according to the time scales. The closed-loop obtained at one stage is the plant for the next stage. In this way a hierarchical and sequential synthesis is possible, with, at each level, account for the faster controls of lower levels and with minimal risk of parasitic dynamic interactions. Notice also that this strategy is compliant with actual organization of controls in power systems (structured in primary/secondary layers) and opens the way of direct integration of PPMs into existing power systems controls and market mechanisms.

\begin{figure}[H]
    \centering
    \includegraphics[width=0.45\textwidth,center]{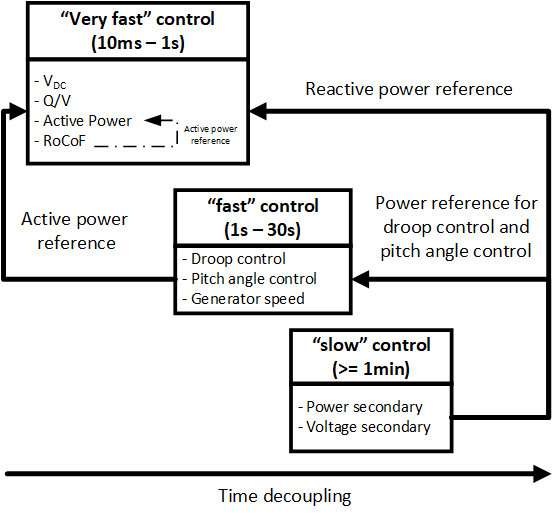}
    \caption{Time decoupling structure for grid connected PPMs} 
    \label{fig:Time_decoupling}   
\end{figure}

\begin{figure}[H]
    \centering
    \includegraphics[width=0.4\textwidth,center]{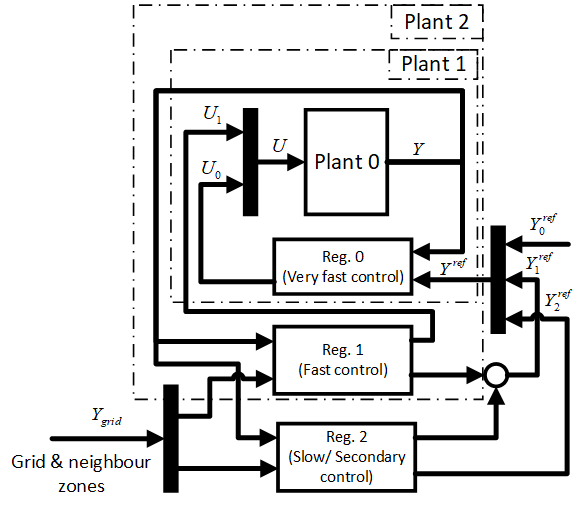}
    \caption{Time and space separation of the control of the grid connected PPMs} 
    \label{fig:Time_and_space_separation}   
\end{figure}

\subsection{A new concept for the control model}

Separation of the V and f dynamics mentioned in Section \ref{section_SystemObjectives} led till now to 2 different models in power systems community: for V regulation, the Single Machine Infinite Bus (SMIB) model given in Fig \ref{fig:V_control_model} and for f regulation, the one in Fig. \ref{fig:f_control_model} (e.g., \cite{Kundur}, \cite{Bevrani}). As explained before, this hypothesis of separation of dynamics is no longer valid in power systems with high RES penetration. Indeed, in case of power converters, V and f dynamics are mixed and both fast. As a consequence, V and f controls should be synthesized together, i.e., in coordination. Advanced methods of control need a \textit{control model}, i.e., a low dimensional model of the plant to be controlled. This is extracted from the \textit{full model} of the plant and should preserve the dynamics of interest. The model we propose here is given in Fig. \ref{fig:control_model_1}. It consists of the full model of the PPM to be controlled, an equivalent AC line of reactance $X_\infty$ and a Grid Dynamic Equivalent (bloc GDE) which provides the grid frequency. The line $X_\infty$ plays the same role as in classic SMIB model in Fig. \ref{fig:V_control_model}. This accounts for the grid short-circuit power at PPM connection bus A and it is computed in a standard way (see, e.g., \cite{12x}). The difference with the SMIB is that electrical frequency of bus B is not fixed at the nominal grid frequency (50 or 60Hz), but given by GDE by the following swing equation

\begin{equation}\label{GDEswing}
	2H\frac{d\omega_f}{dt}=P_G-P_L-D_u\Delta \omega_f
\end{equation}

and the three-phase voltage dynamics

\begin{equation}\label{GDEvoltage}
\begin{array}{l}
\begin{aligned}
	{V_B}^a=Vsin(\theta_f)\\
	{V_B}^b=Vsin(\theta_f-2\frac{\pi}{3})\\
	{V_B}^c=Vsin(\theta_f+2\frac{\pi}{3}).
\end{aligned}
	\end{array}
\end{equation}

H is the equivalent inertia of the rest of the system (in which the controlled PPM is inserted), $P_G$ is the global active power produced in the rest of the system and $P_L$ corresponds to the global load of the system. $P_m$ is a constant input for the PPM control problem. Dynamics (\ref{GDEswing}) is stabilized by the damping factor $D_u$ and a simple integrator for deviation of $\omega_f$ from the nominal grid frequency if a secondary frequency control is considered. H is computed by classic equivalencing methods (e.g., \cite{Kundur}) used in frequency studies. The difference between GDE and model in Fig. \ref{fig:f_control_model} is that GDE captures also voltage dynamics in compliance with the f ones. 

\begin{figure}[H]
    \centering
    \includegraphics[width=0.3\textwidth,center]{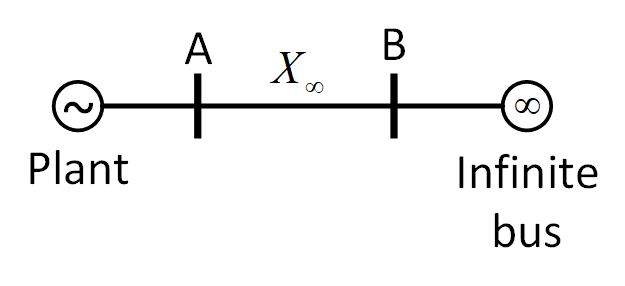}
    \caption{Voltage control model: SMIB} 
    \label{fig:V_control_model}   
\end{figure}

\begin{figure}[H]
    \centering
    \includegraphics[width=0.4\textwidth,center]{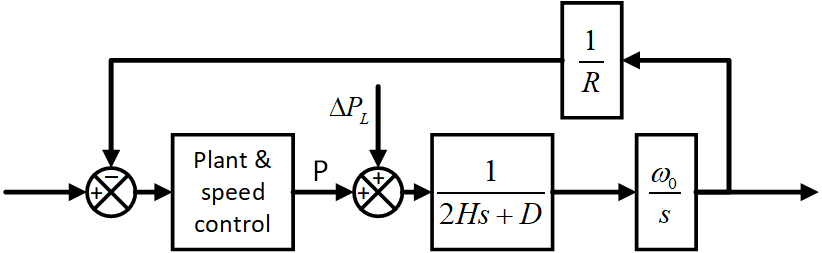}
    \caption{Frequency control model} 
    \label{fig:f_control_model}   
\end{figure}

\begin{figure}[H]
    \centering
    \includegraphics[width=0.3\textwidth,center]{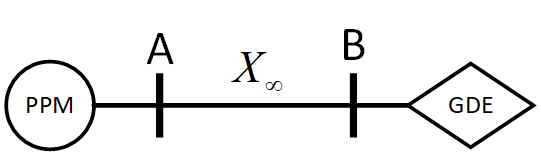}
    \caption{New control model for single PPM} 
    \label{fig:control_model_1}   
\end{figure}


\begin{figure}[H]
    \centering
    \begin{minipage}[b][][b]{0.2\textwidth}
    	\centering
    	\includegraphics[width=1\textwidth,center]{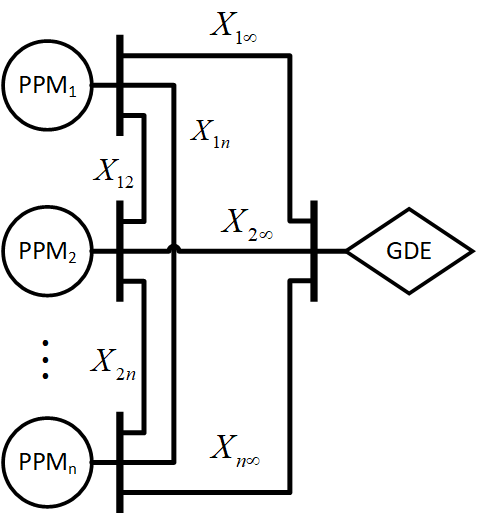}
    \end{minipage}\hfill
    \begin{minipage}[b][][b]{0.24\textwidth}
    	\centering
    	\includegraphics[width=1\textwidth,center]{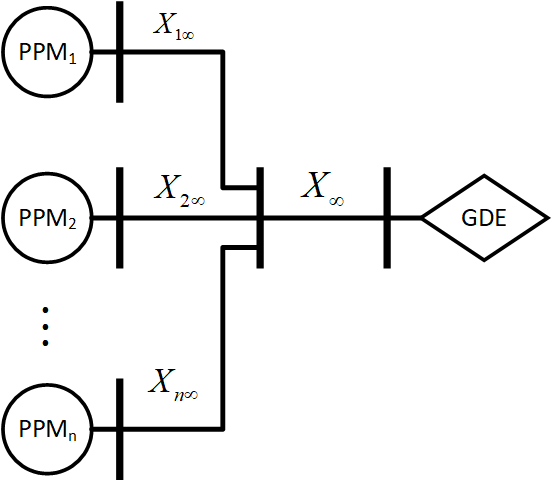}
    \end{minipage}
		\caption{New control model for multiple PPMs}
		\label{controlModelMulti}
\end{figure}

Notice also that a common practice is to use a large classic synchronous generator as equivalent for the rest of the system. Indeed, one can thus chose its inertia constant H close to the inertia of the rest of the system. However, this kind of equivalent is not suitable for at least 2 reasons: first, as the model of a physical synchrounous generator is used, one should also add voltage and frequency regulations for this machine. The latter ones have no significance in the model and their parameters would influence the resulting dynamics. Next, the order of such a model (the dimension of its state) is much higher than the one of the proposed GDE (which is 1). 

If several PPMs have to be controlled, the structure of the control model is one of the 2 equivalent ones given in Fig. \ref{controlModelMulti}.

\section{Application to wind PPMs}

Control strategy presented in the above section will be here developed for 2 wind PPMs based on Permanent Magnet Synchronous Generators (PMSGs). 

\subsection{Synthesis of the control model}
The control model of the first layer of regulation (i.e., Plant 0 in Fig. \ref{fig:Time_and_space_separation}) is thus the one in Fig. \ref{controlModelMulti} with n=2 and the PPM of Fig. \ref{fig:PMSG}, i.e., the one in Fig. \ref{fig:PMSGs_and_grid}.

\begin{figure}[H]
    \centering
    \includegraphics[scale=0.37]{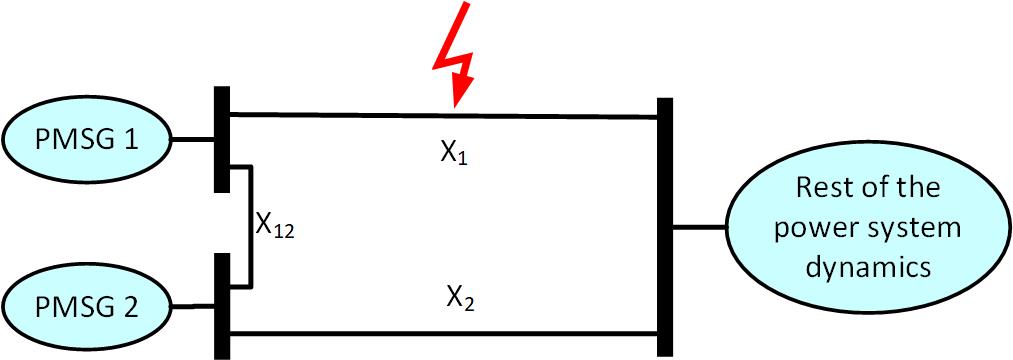}
    \caption{Two PMSGs connected to the grid}
    \label{fig:PMSGs_and_grid}
\end{figure}

\begin{figure}[H]
    \centering
    \includegraphics[scale=0.37]{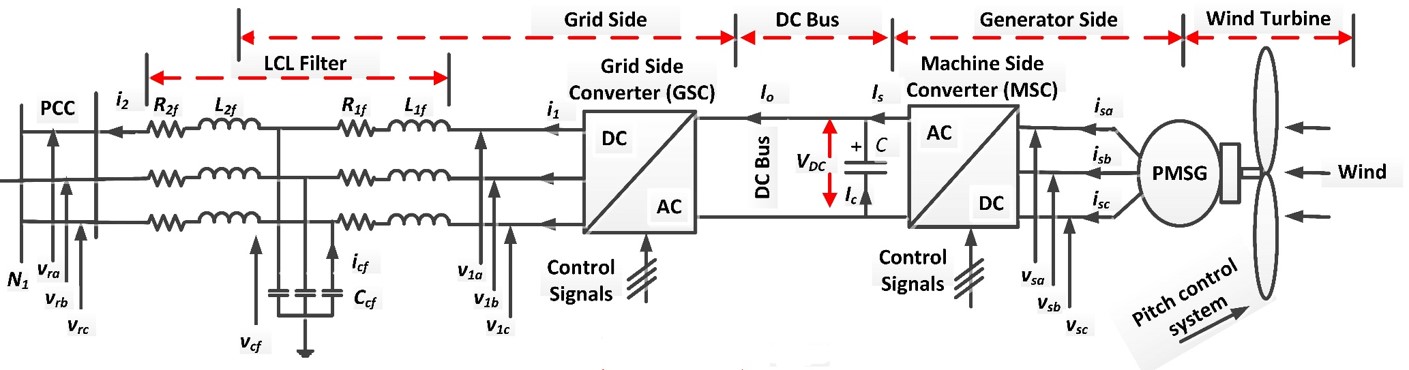}
    \caption{Wind PMSG based PPM}
    \label{fig:PMSG}
\end{figure}

Using the nonlinear model in \cite{1e} for the 2 PMSGs, one obtains 

\begin{equation}\label{eq:nonlinear_PMSG1}
{\tiny
\left\{ 
\begin{array}{l}
\begin{aligned}
\frac{{d{i_{11d}}}}{{dt}} =&  - {\omega _{fb}}\frac{{{R_{11f}}{i_{11d}}}}{{{L_{11f}}}} + {\omega _{fb}}{\omega _f}{i_{11q}} - {\omega _{fb}}\frac{{{v_{cf1d}}}}{{{L_{11f}}}} + \frac{1}{2}{\omega _{fb}}\frac{{{\beta _{11d}}{v_{DC1}}}}{{{L_{11f}}}}
\end{aligned}\\
\begin{aligned}
\frac{{d{i_{11q}}}}{{dt}} =&  - {\omega _{fb}}\frac{{{R_{11f}}{i_{11q}}}}{{{L_{11f}}}} - {\omega _{fb}}{\omega _f}{i_{11d}} - {\omega _{fb}}\frac{{{v_{cf1q}}}}{{{L_{11f}}}} + \frac{1}{2}{\omega _{fb}}\frac{{{\beta _{11q}}{v_{DC1}}}}{{{L_{11f}}}}
\end{aligned}\\
\begin{aligned}
&\frac{{d{i_{21d}}}}{{dt}} =  - {\omega _{fb}}\frac{{{R_{21f}}{i_{21d}}}}{{{L_{21f}}}} + {\omega _{fb}}{\omega _f}{i_{21q}} + {\omega _{fb}}\frac{{{v_{cf1d}}}}{{{L_{21f}}}}\\
& - {\omega _{fb}}\frac{1}{{{L_{21f}}}}\frac{1}{{{n_1}}}\left( {{a_{r1d}}{i_{21d}} + {b_{r1d}}{i_{21q}} + {c_{r1d}}{i_{22d}} + {d_{r1d}}{i_{22q}} + {e_{r1d}}V} \right)
\end{aligned}\\
\begin{aligned}
&\frac{{d{i_{21q}}}}{{dt}} =  - {\omega _{fb}}\frac{{{R_{21f}}{i_{21q}}}}{{{L_{21f}}}} - {\omega _{fb}}{\omega _f}{i_{21d}} + {\omega _{fb}}\frac{{{v_{cf1q}}}}{{{L_{21f}}}}\\
& - {\omega _{fb}}\frac{1}{{{L_{21f}}}}\frac{1}{{{n_1}}}\left( {{a_{r1q}}{i_{21d}} + {b_{r1q}}{i_{21q}} + {c_{r1q}}{i_{22d}} + {d_{r1q}}{i_{22q}} + {e_{r1q}}V} \right)
\end{aligned}\\
\begin{aligned}
\frac{{d{v_{cf1d}}}}{{dt}} =& {\omega _{fb}}\frac{{{i_{11d}}}}{{{C_{f1}}}} - {\omega _{fb}}\frac{{{i_{21d}}}}{{{C_{f1}}}} + {\omega _{fb}}{\omega _f}{v_{cf1q}}
\end{aligned}\\
\begin{aligned}
\frac{{d{v_{cf1q}}}}{{dt}} =& {\omega _{fb}}\frac{{{i_{11q}}}}{{{C_{f1}}}} - {\omega _{fb}}\frac{{{i_{21q}}}}{{{C_{f1}}}} - {\omega _{fb}}{\omega _f}{v_{cf1d}}
\end{aligned}\\
\begin{aligned}
\frac{{d{i_{s1d}}}}{{dt}} =&  - {\omega _{fb}}\frac{{{R_{s1}}{i_{s1d}}}}{{{L_{s1}}}} + {\Omega _b}{p_1}{\Omega _1}{i_{s1q}} - \frac{1}{2}{\omega _{fb}}\frac{{{\beta _{s1d}}{v_{DC1}}}}{{{L_{s1}}}}
\end{aligned}\\
\begin{aligned}
\frac{{d{i_{s1q}}}}{{dt}} =&  - {\omega _{fb}}\frac{{{R_{s1}}{i_{s1q}}}}{{{L_{s1}}}} - {\Omega _b}{p_1}{\Omega _1}{i_{s1d}} + {\Omega _b}\frac{{{p_1}{\Psi _{f1}}{\Omega _1}}}{{{L_{s1}}}} - \frac{1}{2}{\omega _{fb}}\frac{{{\beta _{s1q}}{v_{DC1}}}}{{{L_{s1}}}}
\end{aligned}\\
\begin{aligned}
\frac{{d{v_{DC1}}}}{{dt}} =&  - {\omega _{fb}}\frac{3}{{4{C_1}}}{\beta _{11d}}{i_{11d}} - {\omega _{fb}}\frac{3}{{4{C_1}}}{\beta _{11q}}{i_{11q}} + {\omega _{fb}}\frac{3}{{4{C_1}}}{\beta _{s1d}}{i_{s1d}}\\
& + {\omega _{fb}}\frac{3}{{4{C_1}}}{\beta _{s1q}}{i_{s1q}}
\end{aligned}\\
\begin{aligned}
\frac{{d{\Omega _1}}}{{dt}} =& \frac{1}{{2{H_1}}}\frac{{{P_{m1}}}}{{{\Omega _1}}} - \frac{1}{{2{H_1}}}\frac{{{\Omega _b}}}{{{\omega _{fb}}}}\frac{3}{2}{p_1}{\Psi _{f1}}{i_{s1q}} - \frac{1}{{2{H_1}}}{D_{L1}}{\Omega _1}
\end{aligned}
\end{array} 
\right.
}
\end{equation}

\begin{equation}\label{eq:nonlinear_PMSG2}
{\tiny
\left\{
\begin{array}{l}
\begin{aligned}
\frac{{d{i_{12d}}}}{{dt}} =&  - {\omega _{fb}}\frac{{{R_{12f}}{i_{11d}}}}{{{L_{12f}}}} + {\omega _{fb}}{\omega _f}{i_{12q}} - {\omega _{fb}}\frac{{{v_{cf2d}}}}{{{L_{12f}}}} + \frac{1}{2}{\omega _{fb}}\frac{{{\beta _{12d}}{v_{DC2}}}}{{{L_{12f}}}}
\end{aligned}\\
\begin{aligned}
\frac{{d{i_{12q}}}}{{dt}} =&  - {\omega _{fb}}\frac{{{R_{12f}}{i_{11q}}}}{{{L_{12f}}}} - {\omega _{fb}}{\omega _f}{i_{12d}} - {\omega _{fb}}\frac{{{v_{cf2q}}}}{{{L_{12f}}}} + \frac{1}{2}{\omega _{fb}}\frac{{{\beta _{12q}}{v_{DC2}}}}{{{L_{12f}}}}
\end{aligned}\\
\begin{aligned}
&\frac{{d{i_{22d}}}}{{dt}} =  - {\omega _{fb}}\frac{{{R_{22f}}{i_{22d}}}}{{{L_{22f}}}} + {\omega _{fb}}{\omega _f}{i_{22q}} + {\omega _{fb}}\frac{{{v_{cf2d}}}}{{{L_{22f}}}}\\
& - {\omega _{fb}}\frac{1}{{{L_{22f}}}}\frac{1}{{{n_2}}}\left( {{a_{r2d}}{i_{21d}} + {b_{r2d}}{i_{21q}} + {c_{r2d}}{i_{22d}} + {d_{r2d}}{i_{22q}} + {e_{r2d}}V} \right)
\end{aligned}\\
\begin{aligned}
&\frac{{d{i_{22q}}}}{{dt}} =  - {\omega _{fb}}\frac{{{R_{22f}}{i_{22q}}}}{{{L_{22f}}}} - {\omega _{fb}}{\omega _f}{i_{22d}} + {\omega _{fb}}\frac{{{v_{cf2q}}}}{{{L_{22f}}}}\\
& - {\omega _{fb}}\frac{1}{{{L_{22f}}}}\frac{1}{{{n_2}}}\left( {{a_{r2q}}{i_{21d}} + {b_{r2q}}{i_{21q}} + {c_{r2q}}{i_{22d}} + {d_{r2q}}{i_{22q}} + {e_{r2q}}V} \right)
\end{aligned}\\
\begin{aligned}
\frac{{d{v_{cf2d}}}}{{dt}} =& {\omega _{fb}}\frac{{{i_{12d}}}}{{{C_{f2}}}} - {\omega _{fb}}\frac{{{i_{22d}}}}{{{C_{f2}}}} + {\omega _{fb}}{\omega _f}{v_{cf2q}}
\end{aligned}\\
\begin{aligned}
\frac{{d{v_{cf2q}}}}{{dt}} =& {\omega _{fb}}\frac{{{i_{12q}}}}{{{C_{f2}}}} - {\omega _{fb}}\frac{{{i_{22q}}}}{{{C_{f2}}}} - {\omega _{fb}}{\omega _f}{v_{cf2d}}\\
\frac{{d{i_{s2d}}}}{{dt}} =&  - {\omega _{fb}}\frac{{{R_{s2}}{i_{s2d}}}}{{{L_{s2}}}} + {\Omega _b}{p_2}{\Omega _2}{i_{s2q}} - \frac{1}{2}{\omega _{fb}}\frac{{{\beta _{s2d}}{v_{DC2}}}}{{{L_{s2}}}}
\end{aligned}\\
\begin{aligned}
\frac{{d{i_{s2q}}}}{{dt}} =&  - {\omega _{fb}}\frac{{{R_{s2}}{i_{s2q}}}}{{{L_{s2}}}} - {\Omega _b}{p_2}{\Omega _2}{i_{s2d}} + {\Omega _b}\frac{{{p_2}{\Psi _{f2}}{\Omega _2}}}{{{L_{s2}}}} - \frac{1}{2}{\omega _{fb}}\frac{{{\beta _{s2q}}{v_{DC2}}}}{{{L_{s2}}}}
\end{aligned}\\
\begin{aligned}
\frac{{d{v_{DC2}}}}{{dt}} =&  - {\omega _{fb}}\frac{3}{{4{C_2}}}{\beta _{12d}}{i_{12d}} - {\omega _{fb}}\frac{3}{{4{C_2}}}{\beta _{12q}}{i_{12q}} + {\omega _{fb}}\frac{3}{{4{C_2}}}{\beta _{s2d}}{i_{s2d}}\\
& + {\omega _{fb}}\frac{3}{{4{C_2}}}{\beta _{s2q}}{i_{s2q}}
\end{aligned}\\
\begin{aligned}
\frac{{d{\Omega _2}}}{{dt}} =& \frac{1}{{2{H_2}}}\frac{{{P_{m2}}}}{{{\Omega _2}}} - \frac{1}{{2{H_2}}}\frac{{{\Omega _b}}}{{{\omega _{fb}}}}\frac{3}{2}{p_2}{\Psi _{f2}}{i_{s2q}} - \frac{1}{{2{H_2}}}{D_{L2}}{\Omega _2};{D_{L2}} = \frac{{\Omega _b^2{f_2}}}{{{S_b}}}
\end{aligned}
\end{array}
\right.
}
\end{equation}

For the rest of the above presented control model, one should consider (\ref{GDEswing}) with ${D_u} = \frac{{\omega _b^2D}}{{{S_b}}}$,  $P_G=P_m+P_e$, where $P_m$ is the generated power in the rest of the system and 

\begin{equation}\label{eq:grid_dynamics}
{\small
{P_e} = {a_{Pe}}{i_{21d}} + {b_{Pe}}{i_{21q}} + {c_{Pe}}{i_{22d}} + {d_{Pe}}{i_{22q}}
}
\end{equation}

is the power generated by the two PMSGs.

The dynamics of pitch angle of the two PMSGs are totally decoupled from the rest of the PMSGs dynamics \cite{45} - \cite{46}:

\begin{equation}\label{eq:pitch_angle}
\left\{ \begin{array}{l}
\frac{{d{\beta _{PMSG1}}}}{{dt}} =  - \frac{1}{{{\tau _{C1}}}}{\beta _{PMSG1}} + \frac{1}{{{\tau _{C1}}}}{\beta _{PMSG1ref}}\\
\frac{{d{\beta _{PMSG2}}}}{{dt}} =  - \frac{1}{{{\tau _{C2}}}}{\beta _{PMSG2}} + \frac{1}{{{\tau _{C2}}}}{\beta _{PMSG2ref}}
\end{array} \right.
\end{equation}

where

\[\left\{ \begin{array}{l}
{\beta _{\min }} \le {\beta _{PMSG1}},{\beta _{PMSG2}} \le {\beta _{\max }}\\
{\left( {\frac{{d\beta }}{{dt}}} \right)_{\min }} \le \frac{{d{\beta _{PMSG1}}}}{{dt}},\frac{{d{\beta _{PMSG2}}}}{{dt}} \le {\left( {\frac{{d\beta }}{{dt}}} \right)_{\max }}
\end{array} \right.\]

Typically, as in \cite{45} - \cite{46}, these limitations come from the servo motor:

\[\left\{ \begin{array}{l}
{\tau _{C1}} = 0.2\left( s \right);{\tau _{C2}} = 0.2\left( s \right)\\
{\beta _{\min }} = 0\left( {\deg } \right);{\beta _{\max }} = 20\left( {\deg } \right)\\
{\left( {\frac{{d\beta }}{{dt}}} \right)_{\min }} =  - 10\left( {\deg } \right);{\left( {\frac{{d\beta }}{{dt}}} \right)_{\max }} = 10\left( {\deg } \right)
\end{array} \right.\]

The pitch angle will then be controlled separately from the dynamics of the two PMSGs.\\ 

The active power will be calculated at the generator rather than at the converter to avoid the mismatch in generating generator speed reference based on MPPT and deloaded control

\begin{equation}\label{eq:active_power}
\left\{ \begin{gathered}
  {P_1} = \frac{3}{2}\frac{{{\Omega _b}}}{{{\omega _{fb}}}}p{\psi _{f1}}{i_{s1q}}{\Omega _1} \hfill \\
  {P_2} = \frac{3}{2}\frac{{{\Omega _b}}}{{{\omega _{fb}}}}p{\psi _{f2}}{i_{s2q}}{\Omega _2}. \hfill \\ 
\end{gathered}  \right.
\end{equation}

The reactive power will be calculated from the voltages of capacitor in the LCL filter rather than directly from the converters to avoid direct relationship between input and output of the system

\begin{equation}\label{eq:reactive_power}
\left\{ \begin{gathered}
  {Q_1} = {i_{11d}}{v_{cf1q}} - {i_{11q}}{v_{cf1d}} \hfill \\
  {Q_2} = {i_{12d}}{v_{cf2q}} - {i_{12q}}{v_{cf2d}}. \hfill \\ 
\end{gathered}  \right.
\end{equation}

\subsection{Global structure of the proposed control}

According to the time-separation strategy presented in Section \ref{section:time_decoupling_and_fully_coordinated_control_strategy}, to satisfy the control specifications listed in Sections \ref{section_SystemObjectives} and \ref{section:time_decoupling_and_fully_coordinated_control_strategy}, the structure given in Fig. \ref{fig:global_frequency_support} is proposed for the control of 2 wind PPMs. At this stage, droop control and control for RoCoF improvement were inherited from the classic approach presented in Appendices \ref{DroopControl} and \ref{InertiaControl}. They are placed in fast and, respectively, very fast parts of the global scheme. As they are not model-based controls, a state-space model will be constructed (in the next section) only for Plant 0.

\begin{figure}
    \centering
    \includegraphics[scale=0.6]{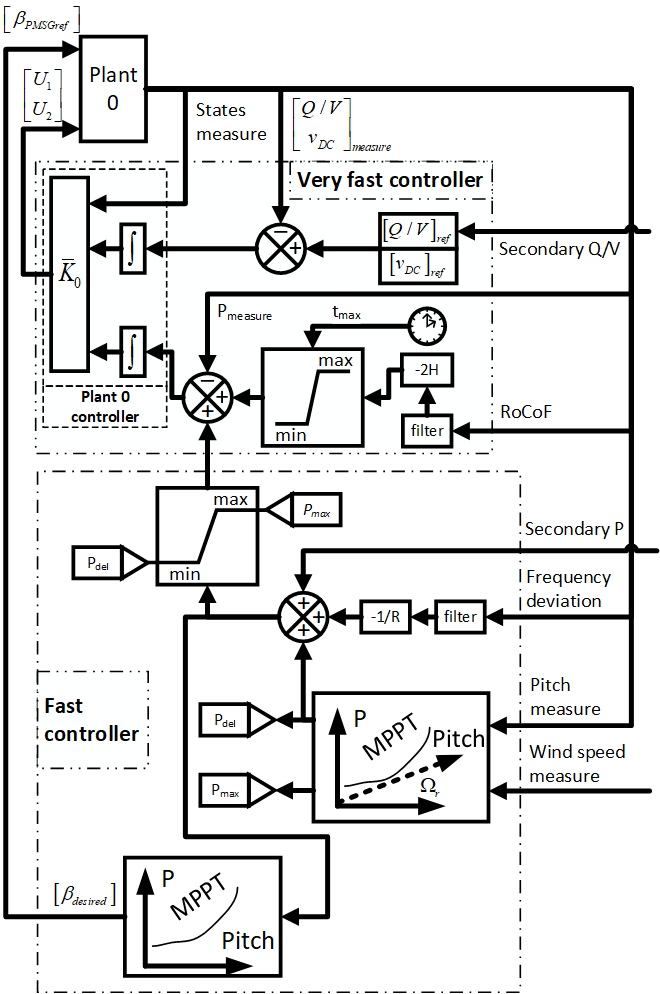}
    \caption{Global control for PMSG with local services and frequency support}
    \label{fig:global_frequency_support}
\end{figure}

\subsection{Augmented system for plant 0 controller}

For the very fast loop control, consider the mixed sensibility control structure as in Fig. \ref{fig:Very_fast_loop}.

\begin{figure}[H]
    \centering
    \includegraphics[scale=0.6]{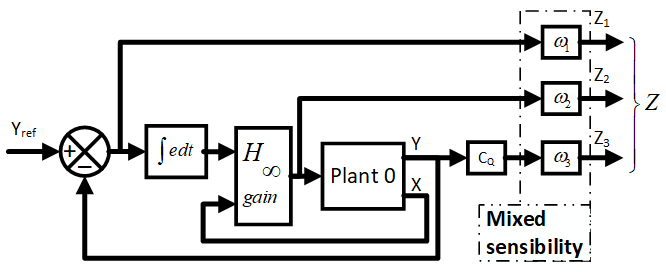}
    \caption{Very fast loop control}
    \label{fig:Very_fast_loop}
\end{figure}

Based on  \eqref{eq:nonlinear_PMSG1}, \eqref{eq:nonlinear_PMSG2},  \eqref{eq:grid_dynamics}, \eqref{eq:pitch_angle}, \eqref{eq:active_power} and \eqref{eq:reactive_power}, the system can be formed in the state-space form with disturbance  

\begin{equation}
\label{eq:PMSG_state_space_system}
\left\{ {\begin{array}{*{20}{l}}
  {{{\dot X}_{e0}} = A{X_{e0}} + {B_1}U + {B_2}W} \\ 
  {{Y_{e0}} = C{X_{e0}}} 
\end{array}} \right.
\end{equation}

The state, input, output and disturbance of the system are

{\small
\[\begin{gathered}
  {X_1} = \left[ {{i_{11d}},{i_{11q}},{i_{21d}},{i_{21q}},{v_{cf1d}},{v_{cf1q}},{i_{s1d}},{i_{s1q}},{v_{DC1}},{\Omega _1}} \right] \hfill \\
  {X_2} = \left[ {{i_{12d}},{i_{12q}},{i_{22d}},{i_{22q}},{v_{cf2d}},{v_{cf2q}},{i_{s2d}},{i_{s2q}},{v_{DC2}},{\Omega _2}} \right] \hfill \\
  {X_{e0}} = {\left[ {\begin{array}{*{20}{c}}
  {{X_1}}&{{X_2}}&{{\omega _f}} 
\end{array}} \right]^T};W = {\left[ {\begin{array}{*{20}{c}}
  {{P_{m1}}}&{{P_{m2}}}&{{P_L}} 
\end{array}} \right]^T} \hfill \\
  {U_1} = \left[ {{\beta _{11d}},{\beta _{11q}},{\beta _{21d}},{\beta _{21q}}} \right];{U_2} = \left[ {{\beta _{12d}},{\beta _{12q}},{\beta _{22d}},{\beta _{22q}}} \right] \hfill \\
  U = \left[ {\begin{array}{*{20}{c}}
  {{U_1}}&{{U_2}} 
\end{array}} \right] \hfill \\
  {Y_{e0}} = {\left[ {\begin{array}{*{20}{l}}
  {{Q_1}}&{{P_1}}&{{v_{DC1}}}&{{Q_2}}&{{P_2}}&{{v_{DC2}}} 
\end{array}} \right]^T} \hfill \\ 
\end{gathered} \]
}

The error between output and the reference is

\begin{equation}
{e_{e0}} =  - C{X_{e0}} + {Y_{e0ref}}
\end{equation}

where the entries of $Y_{e0ref}$ are the reactive power references, the active power references and the DC voltage references.

 Weight $\omega_3$ will improve high-frequencies roll-off in order to cancel the oscillatory effect of LCL filters in Fig. \ref{fig:PMSG}.  Indeed, it is well-known that these filters provide poor damped oscillatory modes at high frequency, around 500Hz in the case of our system. As this frequency is much higher than the one related to the very fast part of our control, roll-off improvement via loop-shaping will be sufficient. Modal analysis (not presented here because of the lack of the space) showed that these modes are observable in the Q measures. $C_Q=\left[ {\begin{array}{*{20}{c}}
  1&0&0&0&0&0 \\ 
  0&0&0&1&0&0 
\end{array}} \right]$ was chosen to target reactive power at the output of this bloc. The structure of the filter $\omega _3$ is standard (see, e.g., \cite{Post}). Let

\begin{equation}
\label{eq:filter_omega_3}
\left\{ \begin{gathered}
  {{\dot X}_{{\omega _3}}} = {A_{{\omega _3}}}{X_{{\omega _3}}} + {B_{{\omega _3}}}{C_Q}CX \hfill \\
  {Y_{{\omega _3}}} = {C_{{\omega _3}}}{X_{{\omega _3}}} + {D_{{\omega _3}}}{C_Q}CX \hfill \\ 
\end{gathered}  \right.
\end{equation}

be its state representation.

The weight $\omega _1$ is usually used to improve reference tracking and disturbance rejection. However, in our case, the integral action structure in Fig. \ref{fig:Very_fast_loop} is perfectly fit for zero-error reference-tracking based on internal model principle. Also, to comply with $H_{\infty}$ optimization to minimize the norm from disturbance to output error, which is the best way for disturbance rejection, $\omega _1$ is chosen as unit gain: $\omega _1 = 1$.

As in our situation there is no particular need for limiting control effort in a particular frequency band, to simplify the control structure, $\omega _2$ will be ignored: $\omega _2 = 0$. 

The \textit{extended} state for the purpose of step reference tracking are

\begin{equation}
{{\bar X}_{e0}} = {\left[ {\begin{array}{*{20}{c}}
  {{X_{e0}}}&{{X_{{\omega _3}}}}&{Int\_e} 
\end{array}} \right]^T};Int\_e = \int {{e_{e0}}dt}
\end{equation}

The resulting extended system is 

\begin{equation}\label{eq:extended_very_fast_system}
\left\{ \begin{gathered}
  {{{\dot {\bar X}}_{e0}}} = {{\bar A}_{e0}}{{\bar X}_{e0}} + {{\bar B}_{1e0}}U + {{\bar B}_{2e0}}{\bar W}_{e0}  \hfill \\
  \overline Y  = {\overline C _{e0}}{{\bar X}_{e0}} + {\overline D _{1e0}}U \hfill \\ 
\end{gathered}  \right.
\end{equation}

where

\[\begin{gathered}
  {{\bar W}_{e0}} = {\left[ {\begin{array}{*{20}{c}}
  {{W^T}}&{{Y_{e0ref}}} 
\end{array}} \right]^T};{{\bar D}_{1e0}} = \left[ {\begin{array}{*{20}{c}}
  0&I \\ 
  0&0 
\end{array}} \right]; \hfill \\
  {{\bar A}_{e0}} = \left[ {\begin{array}{*{20}{c}}
  A&0&0 \\ 
  {{B_{{\omega _3}}}{C_Q}C}&{{A_{{\omega _3}}}}&0 \\ 
  { - C}&0&0 
\end{array}} \right];{{\bar B}_{1e0}} = \left[ {\begin{array}{*{20}{c}}
  {{B_1}} \\ 
  0 \\ 
  0 
\end{array}} \right]; \hfill \\
  {{\bar B}_{2e0}} = \left[ {\begin{array}{*{20}{c}}
  {{B_1}}&0 \\ 
  0&0 \\ 
  0&I 
\end{array}} \right];{C_{e0}} = \left[ {\begin{array}{*{20}{c}}
  { - C}&0&0 \\ 
  {{D_{{\omega _3}}}{C_Q}C}&{{C_{{\omega _3}}}}&0 
\end{array}} \right] \hfill \\ 
\end{gathered} \]

\section{Proposed H-infinity Controller}

The controller synthesis for very fast control loop is given in what follows.\\

Consider the extended linear system \eqref{eq:extended_very_fast_system} of very fast control loop with the state feedback control 
\begin{equation}\label{staticFeedback}
	U = {\bar K}_0 {\bar X}_{e0}.
\end{equation}


\subsection{Controller synthesis}

The influence of the disturbance ${\bar W}_{e0}$ to the output $\bar Y$ is

\begin{equation}\label{transfer}
\bar Y\left( s \right) = G\left( s \right){\bar W}_{e0}\left( s \right),
\end{equation}

where 

\begin{equation}\label{eq:G_s_of_H_inf_linear_system}
{\small
G\left( s \right) = ({{\bar C}_{e0}} + {{\bar D}_{1e0}}{{\bar K}_0}){\left( {sI - \left( {{{\bar A}_{e0}} + {{\bar B}_{1e0}}{{\bar K}_0}} \right)} \right)^{ - 1}}{{\bar B}_{2e0}}
}
\end{equation}

It is obvious that \cite{45}, \cite{41} - \cite{44}

\begin{equation}
{\left\| {\bar Y} \right\|_2} \le {\left\| G \right\|_\infty }{\left\| {{{\bar W}_{e0}}} \right\|_2}
\end{equation}

\begin{problem}\label{pb:1}

For the linear system \eqref{eq:extended_very_fast_system} and \eqref{transfer}, the H-infinity problem is to design a stabilizing static state feedback control law \eqref{staticFeedback} such that

\begin{equation}
{\left\| G \right\|_\infty } < \gamma 
\end{equation}

for a minimum positive scalar $\gamma$.

\end{problem}

\subsection{Derivation of stability and robustness conditions}

The main result for the global asymptotic stability of a globally H-infinity controller  are summarized by the following theorem.

\begin{theorem}\label{theo:1} 

The H-infinity problem has a solution if and only if there exist a matrix $W_P$ and a symmetric positive definite matrix $X_P$ , through the LMI-based minimization such that

\[
\mathop {minimize}\limits_{{X_P}, {W_{P}}} \;\;\gamma 
\]

subject to

\[\left[ {\begin{array}{*{20}{c}}
\Psi &{{{\bar B}_{2e0}}}&{{{\left( {{{\bar C}_{e0}}{X_P} + {{\bar D}_{1e0}}{W_P}} \right)}^T}}\\
{\bar B_{2e0}^T}&{ - \gamma I}&0\\
{{{\bar C}_{e0}}{X_P} + {{\bar D}_{1e0}}{W_P}}&0&{ - \gamma I}
\end{array}} \right] < 0\]

where

\begin{equation}
\Psi  = {\left( {{{\bar A}_{e0}}{X_P} + {{\bar B}_{1e0}}{W_P}} \right)^T} + {{\bar A}_{e0}}{X_P} + {{\bar B}_{1e0}}{W_P}
\end{equation}

In this case, the solution of the problem is

\begin{equation}
{{\bar K}_0} = {W_P}X_P^{ - 1}
\end{equation}
\end{theorem}

\begin{proof}

It is well-known \cite{41} - \cite{44} that for a linear system given by

\begin{equation}
\left\{ \begin{array}{l}
\dot X = AX + BU\\
Y = CX + DU
\end{array} \right.
\end{equation}

or by

\begin{equation}
\left\{ \begin{array}{l}
Y = G\left( s \right)U\\
G\left( s \right) = C{\left( {sI - A} \right)^{ - 1}}B + D
\end{array} \right.
\end{equation}

a solution to the H-infinity problem can be found if and only if there exist a symmetric positive definite matrix $X_p$ , such that

\begin{equation}
\left[ {\begin{array}{*{20}{c}}
{{X_P}{A^T} + A{X_P}}&B&{{X_P}{C^T}}\\
{{B^T}}&{ - \gamma I}&{{D^T}}\\
{C{X_P}}&D&{ - \gamma I}
\end{array}} \right] < 0
\end{equation}

Written for the linear system with disturbance (equation (\ref{eq:extended_very_fast_system})), where $G\left( s \right)$ is in the form of equation (\ref{eq:G_s_of_H_inf_linear_system}), the latter condition yields 

\begin{equation}
\left\{ \begin{array}{l}
{\small
\left[ {\begin{array}{*{20}{c}}
\Psi &{{{\bar B}_{2e0}}}&{{X_P}{{\left( {{{\bar C}_{e0}} + {{\bar D}_{1e0}}{{\bar K}_0}} \right)}^T}}\\
{\bar B_{2e0}^T}&{ - \gamma I}&0\\
{\left( {{{\bar C}_{e0}} + {{\bar D}_{1e0}}{{\bar K}_0}} \right){X_P}}&0&{ - \gamma I}
\end{array}} \right] < 0
}\\
\Psi  = {X_P}{\left( {{{\bar A}_{e0}} + {{\bar B}_{1e0}}{{\bar K}_0}} \right)^T} + \left( {{{\bar A}_{e0}} + {{\bar B}_{1e0}}{{\bar K}_0}} \right){X_P}
\end{array} \right.
\end{equation}

This is equivalent to

\begin{equation}
\left\{ \begin{array}{l}
{\scriptsize
\left[ {\begin{array}{*{20}{c}}
\Psi &{{{\bar B}_{2e0}}}&{{{\left( {{{\bar C}_{e0}}{X_P} + {{\bar D}_{1e0}}{{\bar K}_0}{X_P}} \right)}^T}}\\
{\bar B_{2e0}^T}&{ - \gamma I}&0\\
{{{\bar C}_{e0}}{X_P} + {{\bar D}_{1e0}}{{\bar K}_0}{X_P}}&0&{ - \gamma I}
\end{array}} \right] < 0
}\\
{\footnotesize
\Psi  = {\left( {{{\bar A}_{e0}}{X_P} + {{\bar B}_{1e0}}{{\bar K}_0}{X_P}} \right)^T} + \left( {{{\bar A}_{e0}}{X_P} + {{\bar B}_{1e0}}{{\bar K}_0}{X_P}} \right)
}
\end{array} \right.
\end{equation}

and, by defining ${W_P} = {{\bar K}_0}{X_P}$, the inequality turns into the one of Theorem 1.

\end{proof}

\subsection{Implementation of the controller}

The state-feedback of the above extended system results into a multivariable PI controller
\begin{equation}\label{eq:extended_system_controller}
U = {{\bar K}_0}{{\bar X}_{e0}} = {{\bar K}_0}\left[ {\begin{array}{*{20}{c}}
{{X_{e0}}}\\
{{X_{{\omega _3}}}}\\
{\int {{e_{e0}}dt} }
\end{array}} \right]
\end{equation} 

This is the control structure of Plant 0 controller in Fig. \ref{fig:global_frequency_support}. 

\section{Simulations and results}

For the system in Fig. \ref{fig:PMSGs_and_grid}, simulations are carried out with a full model (detailed models from the standard Matlab libraries for all elements of the systems (generators, lines, converters, loads, ...)) developped in SimPower Toolbox of Matlab. The PMSGs are identical and of 8MW nominal power. The load of the whole power system is $P_L =55000MW$. A limited but visible support is thus expected from PMSGs to grid frequency. The initial operating point for all test scenarios is at generated power $P_e =0.6527pu$ for each PMSG.

To check the PMSGs' control specifications and objectives listed in Sections \ref{section_SystemObjectives} and \ref{section:time_decoupling_and_fully_coordinated_control_strategy}, several scenarios have been constructed:

\begin{itemize}
	\item Scenario 1 - Local services: to test active power, reactive power and DC voltage control, steps are applied on control reference variables. 
	\item Scenario 2 - Voltage services (Q/V): to test voltage services, a three-phase metallic short-circuit is performed in the middle of the AC line $X_1$ in Fig. \ref{fig:PMSGs_and_grid}. The short-circuit starts at t = 100s and it is cleared at t = 100.11s. 
	\item Scenario 3 - Frequency services: to test frequency services, the grid load is step increased by 10Mw at t = 100s and restored at t = 800s. 
	\item Scenario 4 - Comparison between the proposed coordinated control and classic vector control: for that purpose, the classic vector control presented in Appendix \ref{subsectionVectorControl} have been implemented and used. The same short-circuit as in Scenario 2 was considered. 
\end{itemize}	

\subsection{Scenario 1: Local services}        

Fig. \ref{fig:Sim_1_PMSG1_Q1} gives the reactive power response of PMSG1 to steps on $Q_{1_{ref}}$ in both directions (negative and, next, positive). One can notice nominal response times (around 4s for 5$\%$ step tracking) and non oscillatory transient dynamics. Fig. \ref{fig:Sim_1_PMSG1_P1} contains same kind of nominal responses to similar steps on $P_{1_{ref}}$. In Fig. \ref{fig:Sim_1_PMSG1_v_DC1} it is shown the response of the DC voltage control to a change of reference $V_{DC_{1_{ref}}}$ from 2.5pu to 2.625pu. The time response is around 5s and the transient is smooth. The behavior of the second PMSG for local services is similar (Figs. \ref{fig:Sim_1_PMSG2_Q2}, \ref{fig:Sim_1_PMSG2_P2}, \ref{fig:Sim_1_PMSG2_v_DC2}). 

\subsection{Scenario 2: Voltage services (Q/V)}\label{subsectionSecenario2}

Short-circuit responses of PMSG1 are given in Figs. \ref{fig:Sim_2_PMSG1_Q1}, \ref{fig:Sim_2_PMSG1_V_r1} and \ref{fig:Sim_2_PMSG1_w1}. They show a good fault ride through capability, both for voltage and reactive power dynamics. The duration of the transient at the beginning and at the end of short-circuit event is very short, around 4ms and, after, the terminal voltage is stabilized at its desired value. Also, thanks to coordinated control, the impact of the speed of the generator is minimal. Indeed, the generator speed (Fig. \ref{fig:Sim_2_PMSG1_w1}) only experiences a very small deviation in a short time and without any overshoot/undershoot after the clearing of the short-circuit. 

Notice also that, due to the location of the short-circuit (on line $X_1$), its impact on PMSG1 is stronger than on PMSG2. Indeed, one can see in Fig. \ref{fig:Sim_2_PMSG1_V_r1} that the terminal voltage of PMSG1 drops almost to zero. For PMSG2, the effect of short-circuit is insignificant as seen in Fig. \ref{fig:Sim_2_PMSG2_Q2} to Fig. \ref{fig:Sim_2_PMSG2_w2}. In Fig. \ref{fig:Sim_2_PMSG2_Q2}, the reactive power is slightly deviated, around 4$\%$ of its nominal value during the event, and only oscillates around 4ms before coming back to its operating point. In Fig. \ref{fig:Sim_2_PMSG2_V_r2}, the terminal voltage of PMSG2 is maintained at its nominal value, with only a short transient of about 10ms at the beginning and at the end of the event. Moreover, these transients come with a relatively low deviation peak, at only around 4$\%$. In Fig. \ref{fig:Sim_2_PMSG2_w2}, the influence of short-circuit on generator speed is infinitesimal, with only a very small deviation for a very short duration.

\subsection{Scenario 3: Frequency services}

Figs. \ref{fig:Sim_3_PMSG2_P2} and \ref{fig:Sim_3_PMSG1_P1} show very fast increase of the active powers of both PMSGs in response to the simulated load variation of the system given in Fig. \ref{fig:Sim_3_P_L}: they are increased immediately, from deloaded point (0.65pu) to nearly maximum value (0.8pu). This proves effectiveness of the very fast control loop  controller (through inertia control) of Fig. \ref{fig:global_frequency_support}. Indeed, without inertia control, no variation is registered (black curves in the same figures). This frequency support helps to diminish the severe drop of RoCoF as observed in Fig. \ref{fig:Sim_3_RoCoF}. The generator speeds, given in Fig. \ref{fig:Sim_3_PMSG2_w2} and Fig. \ref{fig:Sim_3_PMSG1_w1}, decrease and come close to their optimal values (blue curves/w1\_optimal), which help to maintain the produced power. However, this immediate increase in power can only be maintained during a short laps of time (a couple seconds). After that, the active power of each PMSG tends to come back close to nominal value to ensure machines' stability. 

After a couple of seconds, the fast loop of the proposed control (Fig. \ref{fig:global_frequency_support}) is activated, through MPPT and droop control, to request more power in order to reduce grid frequency deviation (Fig. \ref{fig:Sim_3_wf}). Both generators use their deloding reserves (10$\%$) to increase their powers from deloaded points (0.65pu) to their maximum values of MPPT (0.8pu) as seen in Fig. \ref{fig:Sim_3_PMSG2_P2} and, respectively, Fig. \ref{fig:Sim_3_PMSG1_P1}. As a result, generator speeds and pitch angles come close to optimal values as can be seen in Fig. \ref{fig:Sim_3_PMSG1_w1}, Fig. \ref{fig:Sim_3_PMSG2_w2} and, respectively, Fig. \ref{fig:Sim_3_PMSG1_pitch_angle}, Fig. \ref{fig:Sim_3_PMSG2_pitch_angle} (optimal value for pitch angle is 0 deg). At t = 800s, as the grid load comes back to its nominal value, the active powers for both PMSGs are brought to their nominal values by the very fast and fast controllers of each machine (Fig. \ref{fig:global_frequency_support}). This is the same for grid frequency, generator speeds and pitch angles, as observed in Figs. \ref{fig:Sim_3_wf}, \ref{fig:Sim_3_PMSG1_w1}, \ref{fig:Sim_3_PMSG2_w2}, \ref{fig:Sim_3_PMSG1_pitch_angle}, \ref{fig:Sim_3_PMSG2_pitch_angle}.

\subsection{Scenario 4: Comparison between proposed coordinated control and classic vector control}

One can see in Fig. \ref{fig:Sim_4_PMSG1_v_DC1} and Fig. \ref{fig:Sim_4_PMSG2_v_DC2} that the proposed control provides a better dynamics for $V_{DC}$ than classic vector control. Indeed, transient is less oscillating, which results in shorter response time (around 5 times). This is also the case for the terminal voltage where more sustained oscillations are registered with the vector control (Fig. \ref{fig:Sim_4_PMSG1_V_r1}). 

Same behavior is noticed for PMSG2 (Fig. \ref{fig:Sim_4_PMSG2_V_r2} and \ref{fig:Sim_4_PMSG2_v_DC2}), even though the short-circuit less impacts this machine (as explained in Section \ref{subsectionSecenario2}).


\begin{figure}[H]
    \centering
    \includegraphics[scale=0.14]{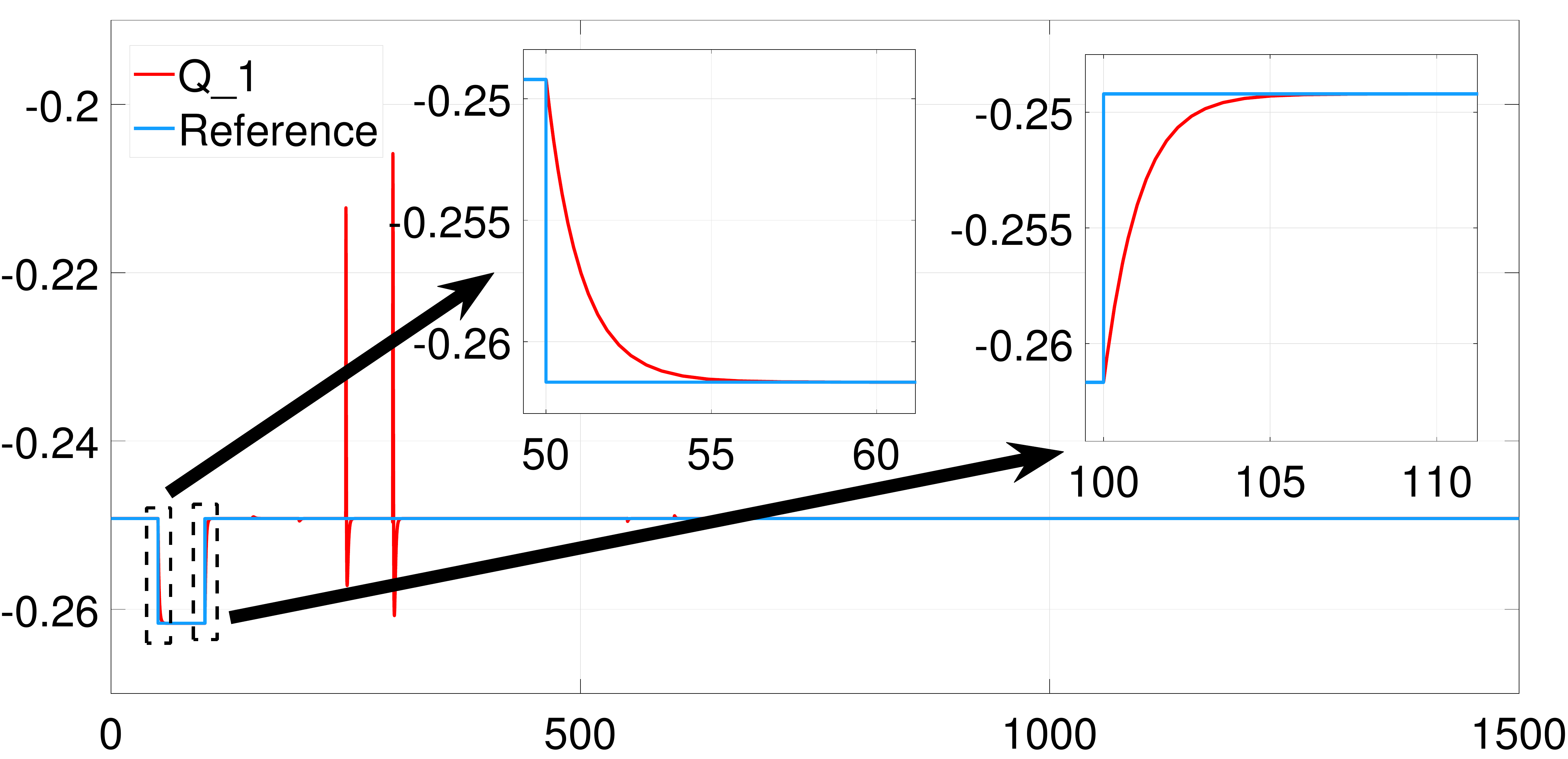}
    \caption{The reactive power (pu) of PMSG1 (local services)}
    \label{fig:Sim_1_PMSG1_Q1}
\end{figure}

\begin{figure}[H]
    \centering
    \includegraphics[scale=0.14]{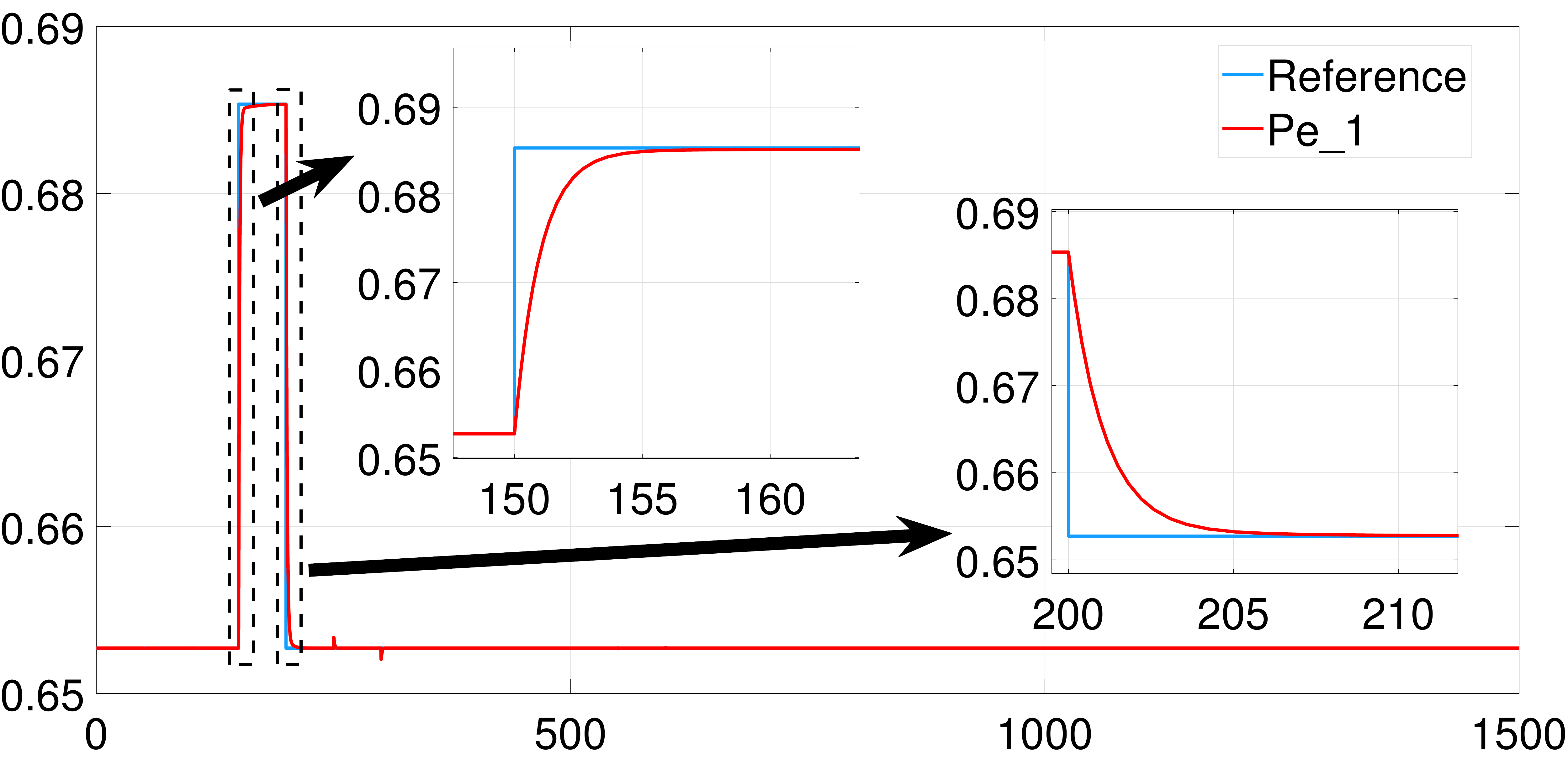}
    \caption{The active power (pu) of PMSG1 (local services)}
    \label{fig:Sim_1_PMSG1_P1}
\end{figure}

\begin{figure}[H]
    \centering
    \includegraphics[scale=0.14]{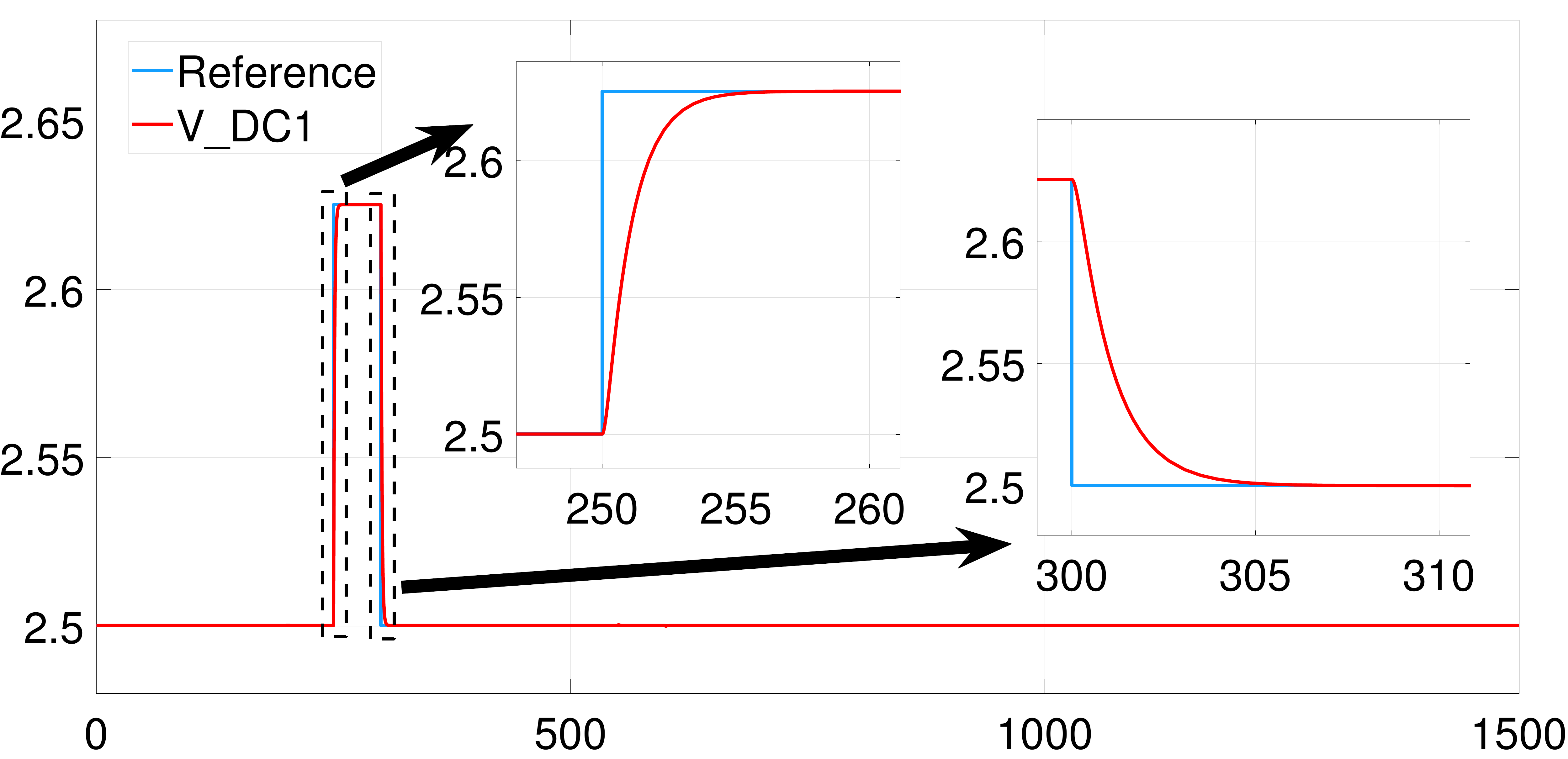}
    \caption{The DC voltage (pu) of PMSG1 (local services)}
    \label{fig:Sim_1_PMSG1_v_DC1}
\end{figure}

\begin{figure}[H]
    \centering
    \includegraphics[scale=0.14]{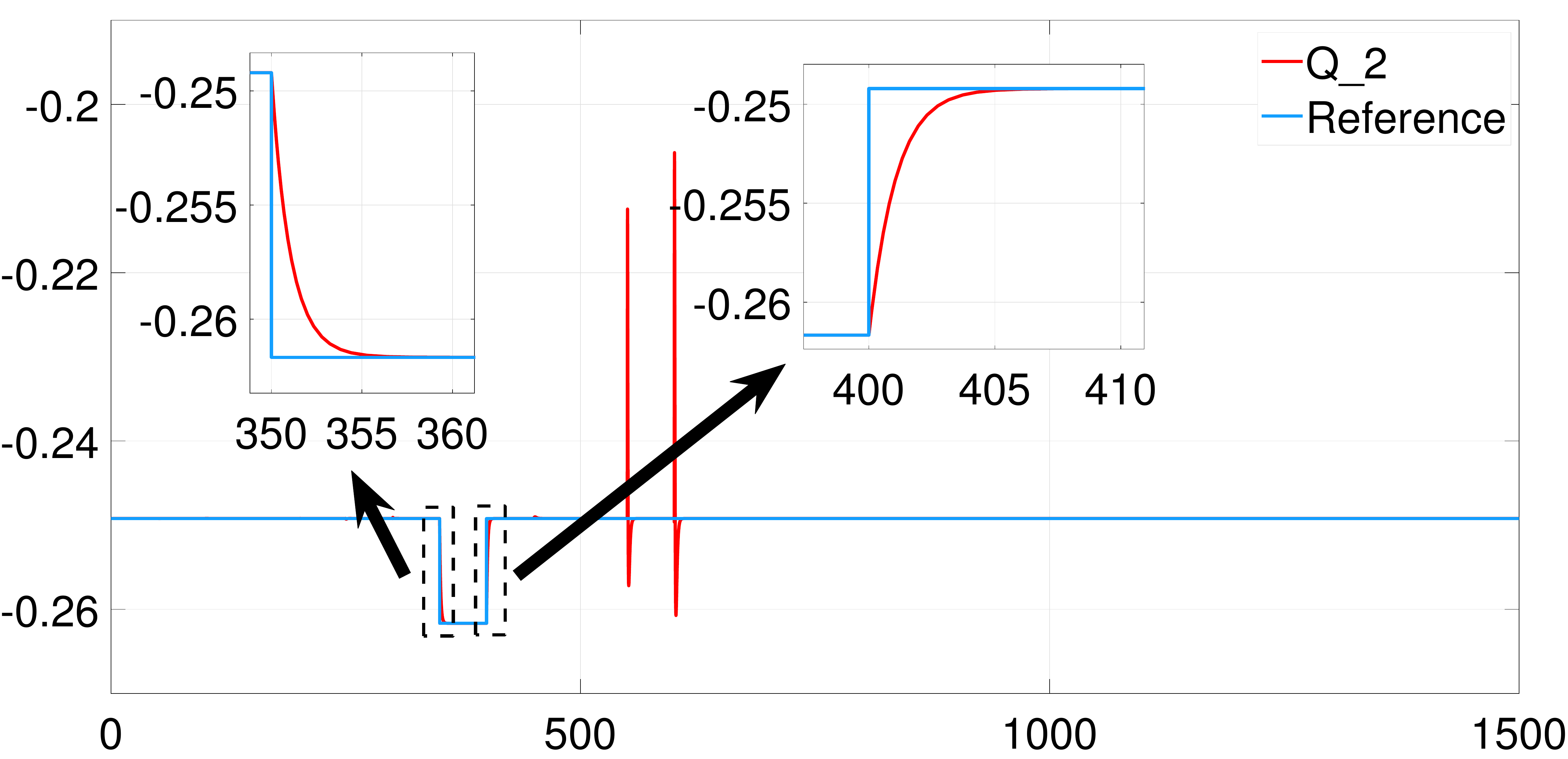}
    \caption{The reactive power (pu) of PMSG2 (local services)}
    \label{fig:Sim_1_PMSG2_Q2}
\end{figure}

\begin{figure}[H]
    \centering
    \includegraphics[scale=0.14]{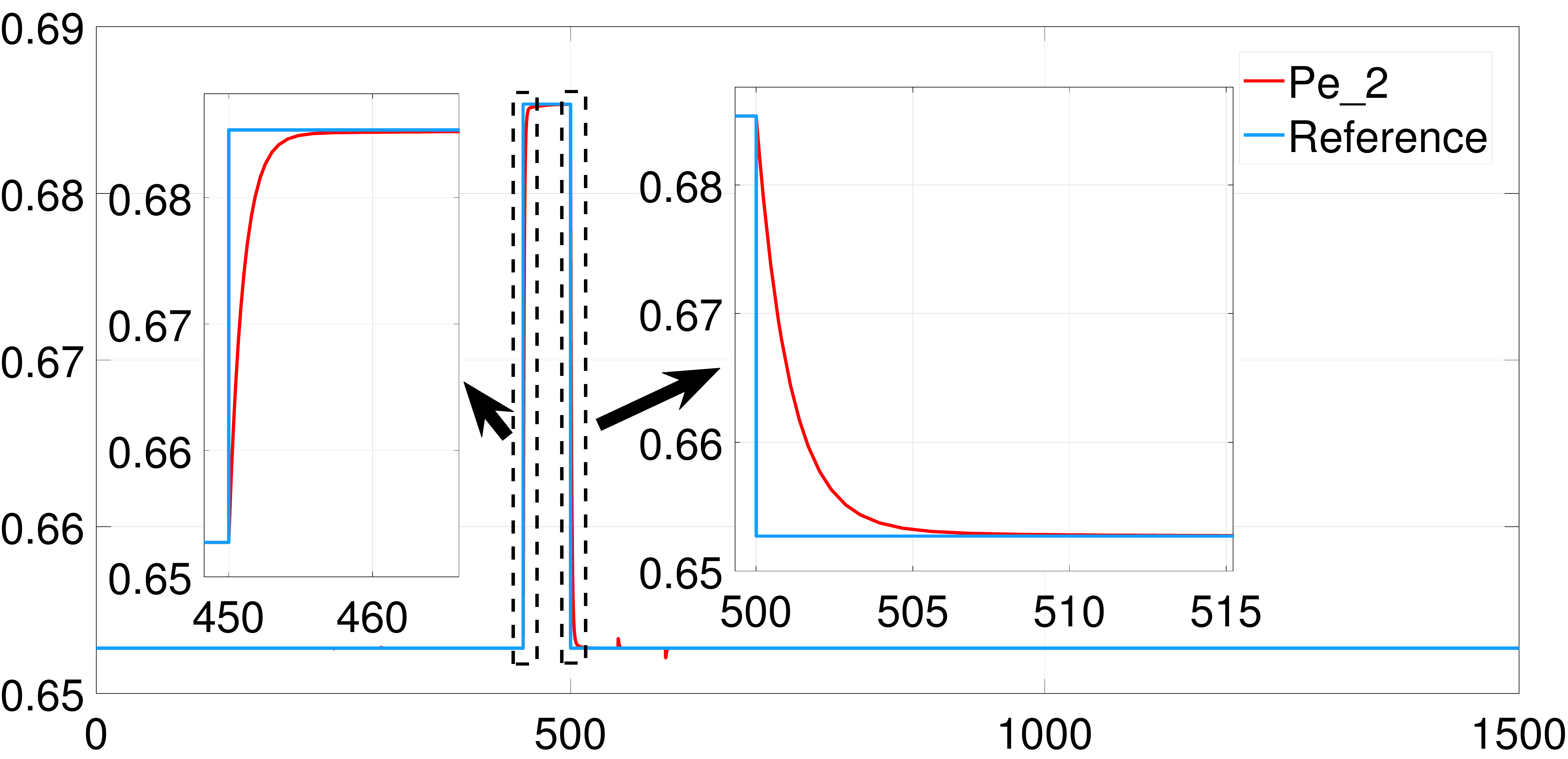}
    \caption{The active power (pu) of PMSG2 (local services)}
    \label{fig:Sim_1_PMSG2_P2}
\end{figure}

\begin{figure}[H]
    \centering
    \includegraphics[scale=0.14]{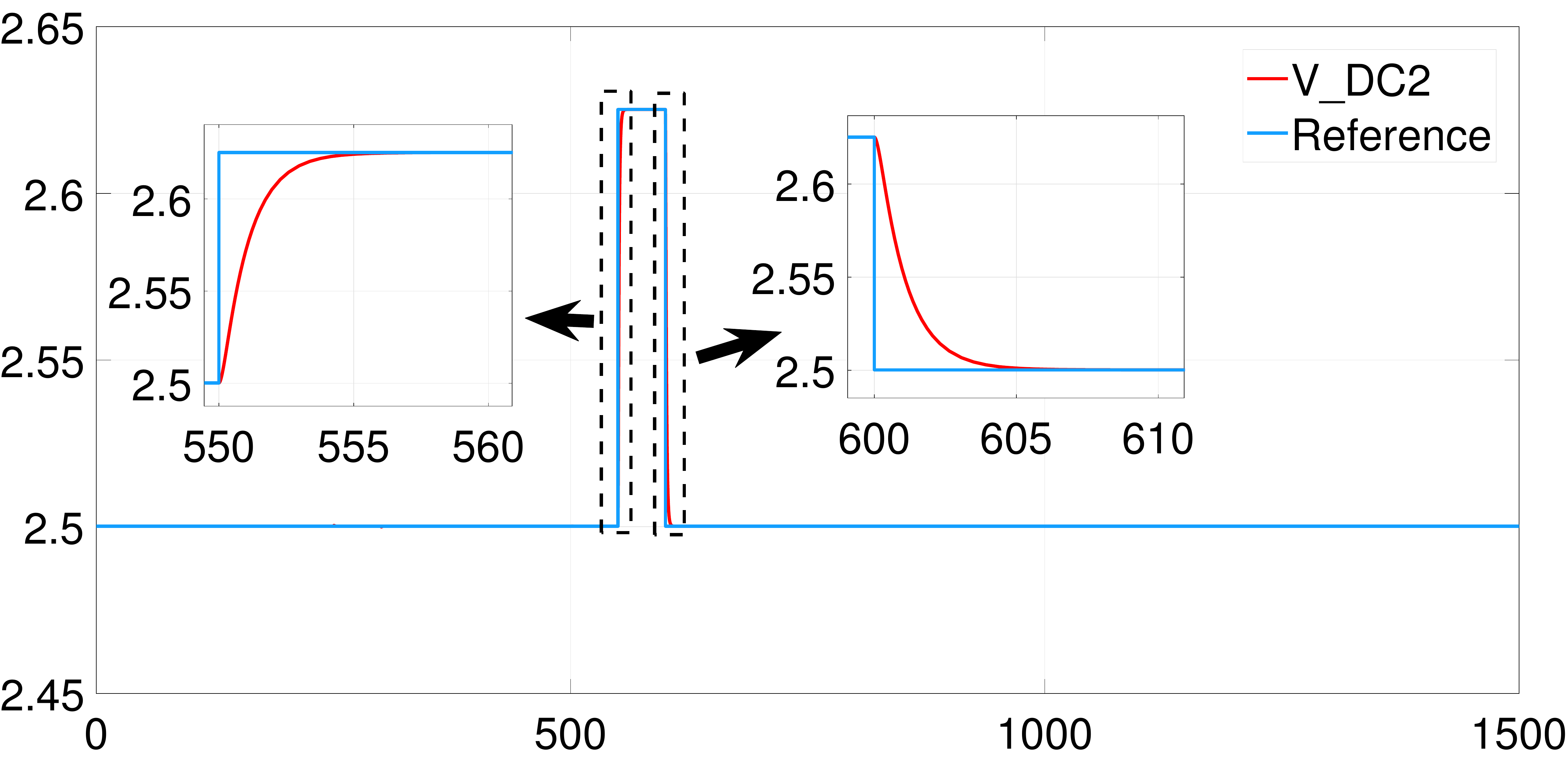}
    \caption{The DC voltage (pu) of PMSG2 (local services)}
    \label{fig:Sim_1_PMSG2_v_DC2}
\end{figure}


\begin{figure}[H]
    \centering
    \includegraphics[scale=0.14]{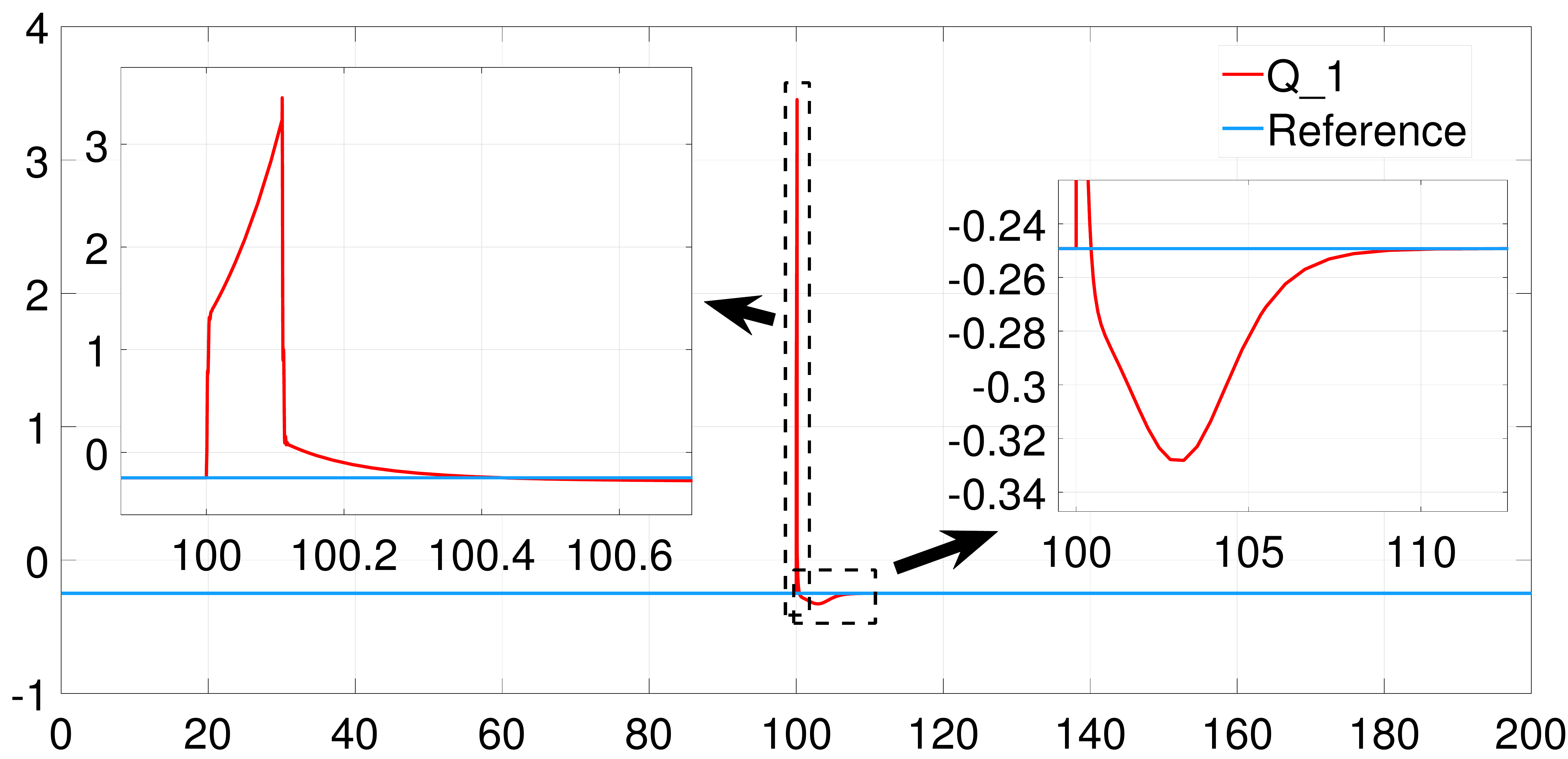}
    \caption{The reactive power (pu) of PMSG1 (Voltage service)}
    \label{fig:Sim_2_PMSG1_Q1}
\end{figure}

\begin{figure}[H]
    \centering
    \includegraphics[scale=0.14]{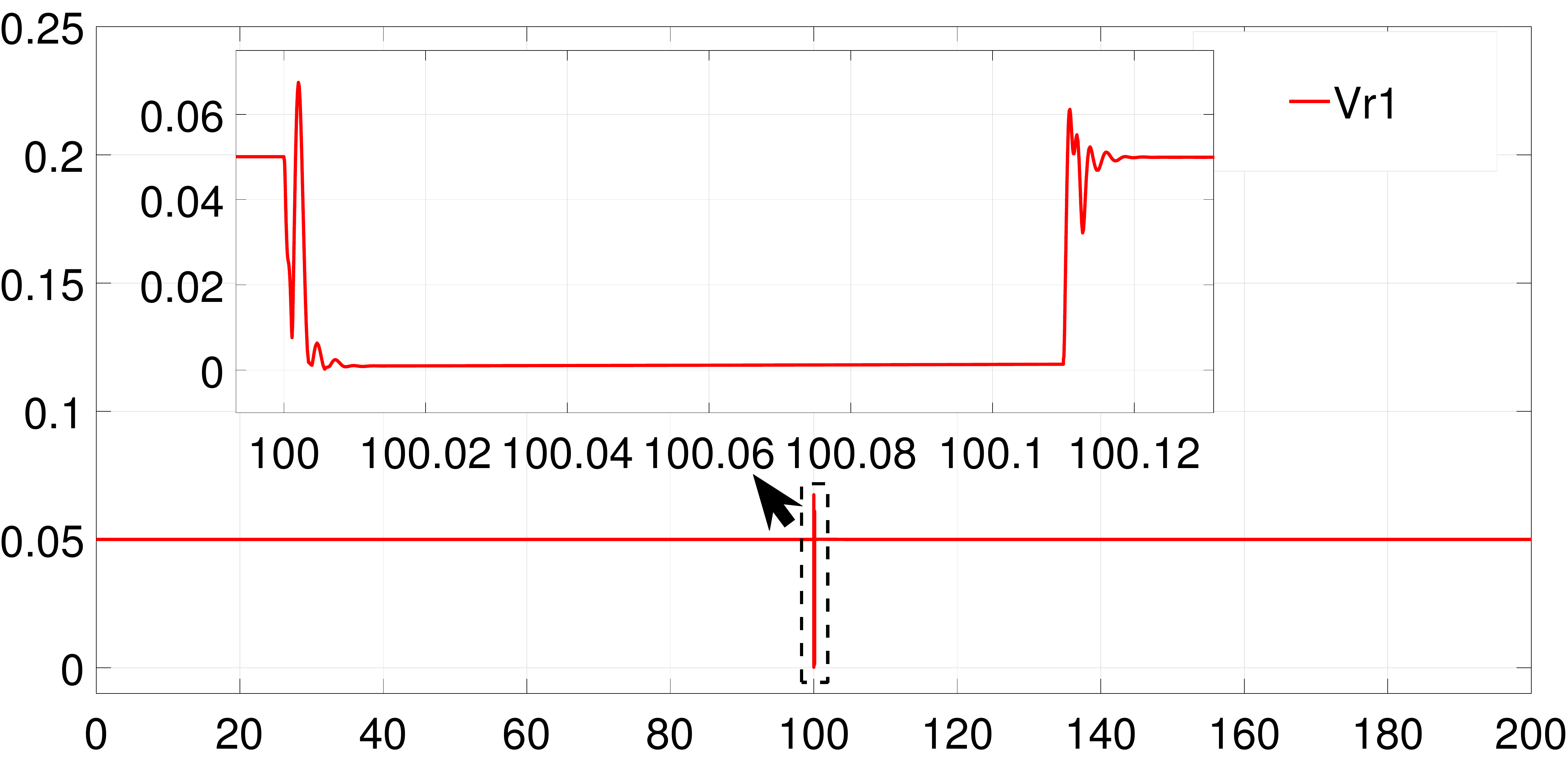}
    \caption{The AC terminal voltage $V_{r1}$ (pu) of PMSG1 (Voltage service)}
    \label{fig:Sim_2_PMSG1_V_r1}
\end{figure}

\begin{figure}[H]
    \centering
    \includegraphics[scale=0.14]{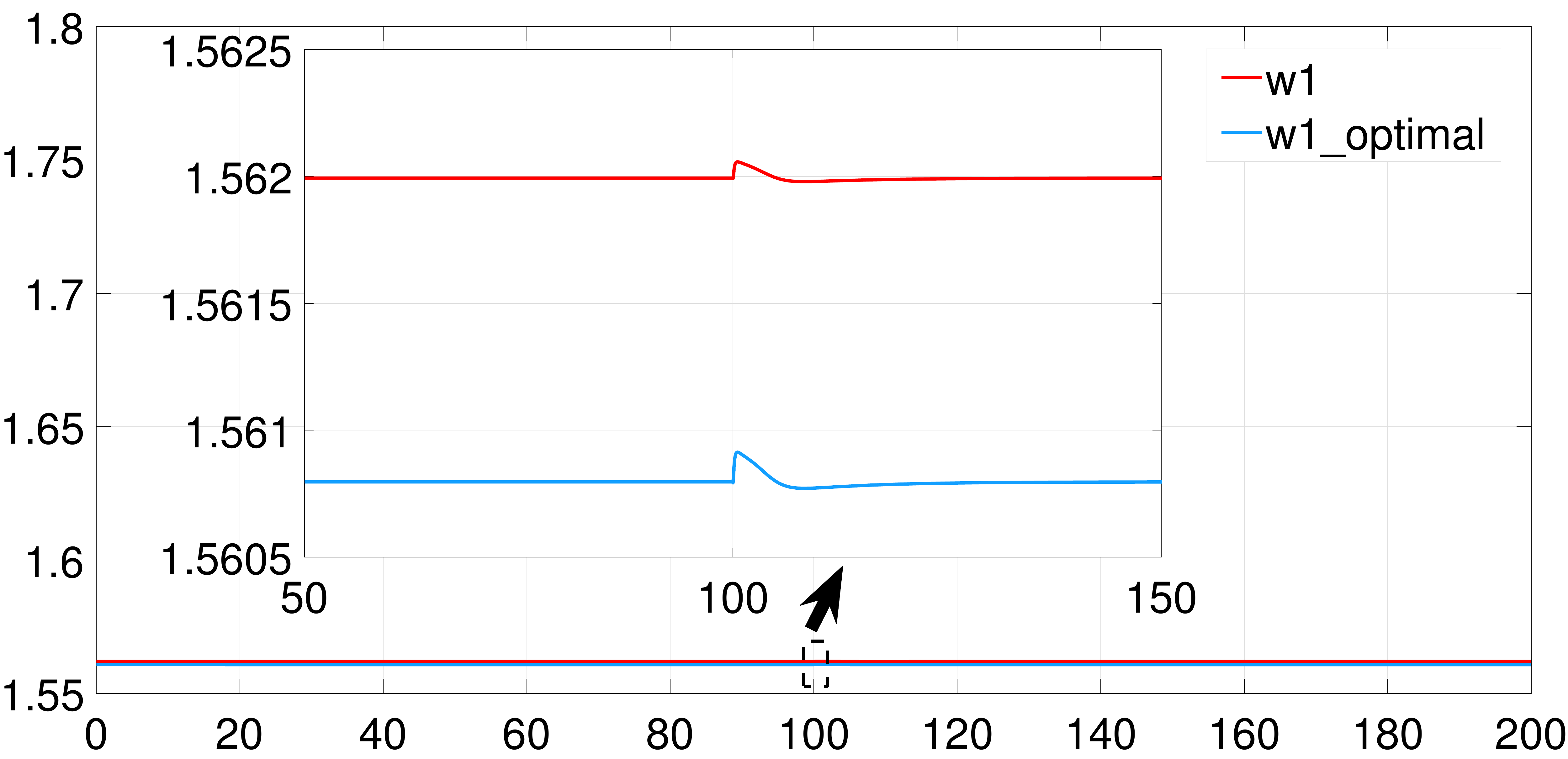}
    \caption{The generator speed (pu) of PMSG1 (Voltage service)}
    \label{fig:Sim_2_PMSG1_w1}
\end{figure}

\begin{figure}[H]
    \centering
    \includegraphics[scale=0.14]{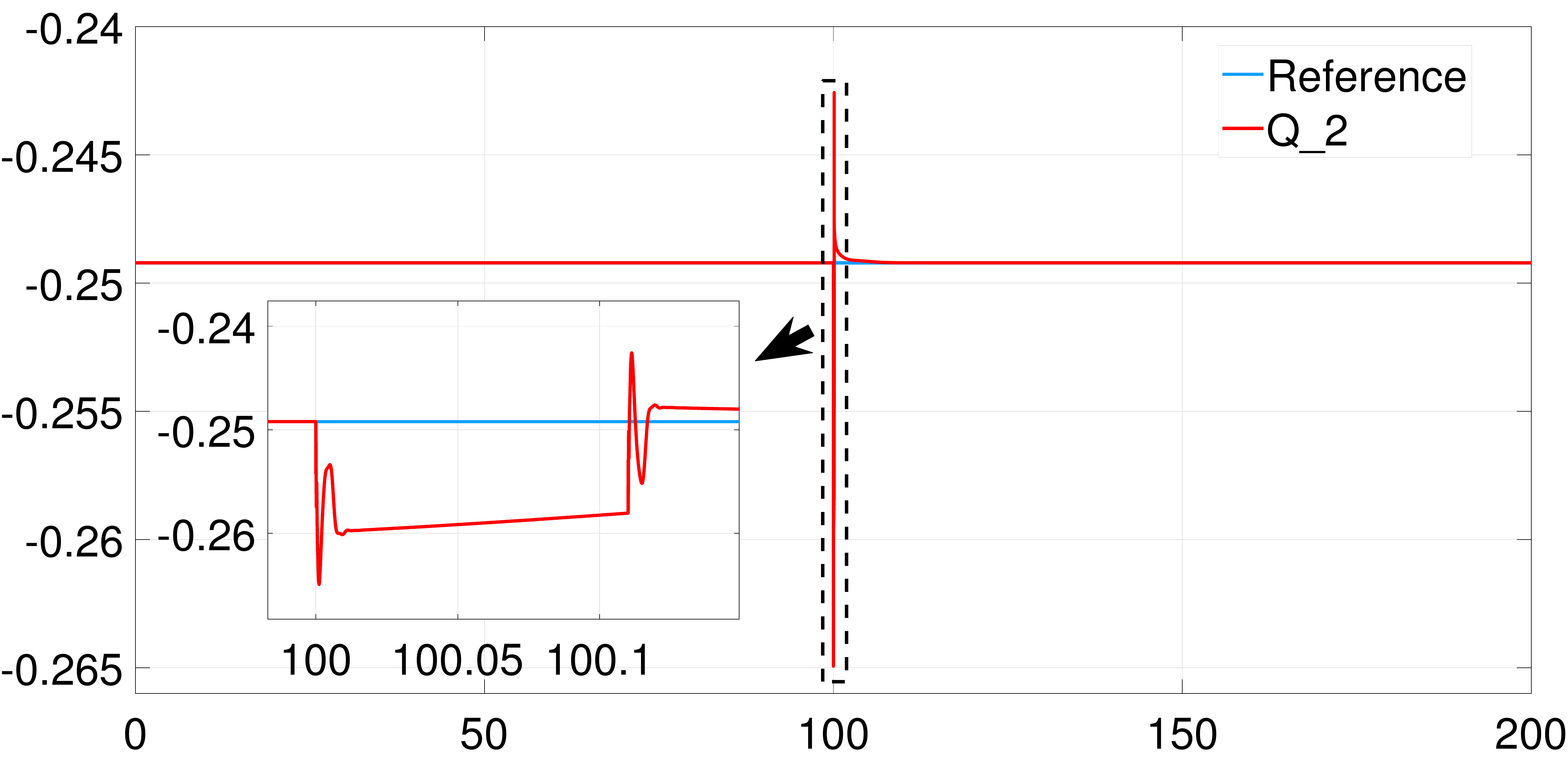}
    \caption{The reactive power (pu) of PMSG2 (Voltage service)}
    \label{fig:Sim_2_PMSG2_Q2}
\end{figure}

\begin{figure}[H]
    \centering
    \includegraphics[scale=0.14]{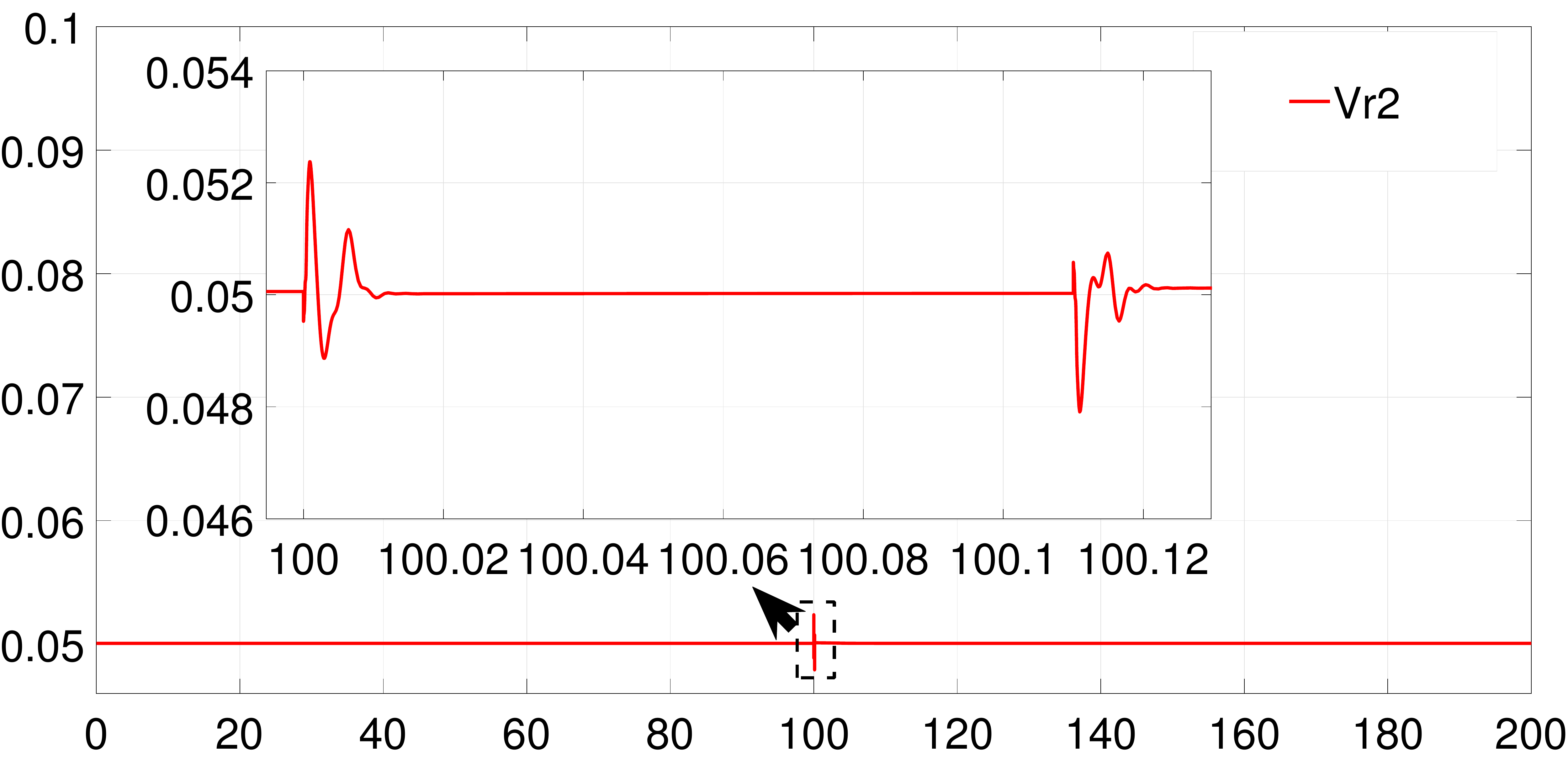}
    \caption{The AC terminal voltage $V_{r2}$ (pu) of PMSG2 (Voltage service)}
    \label{fig:Sim_2_PMSG2_V_r2}
\end{figure}

\begin{figure}[H]
    \centering
    \includegraphics[scale=0.14]{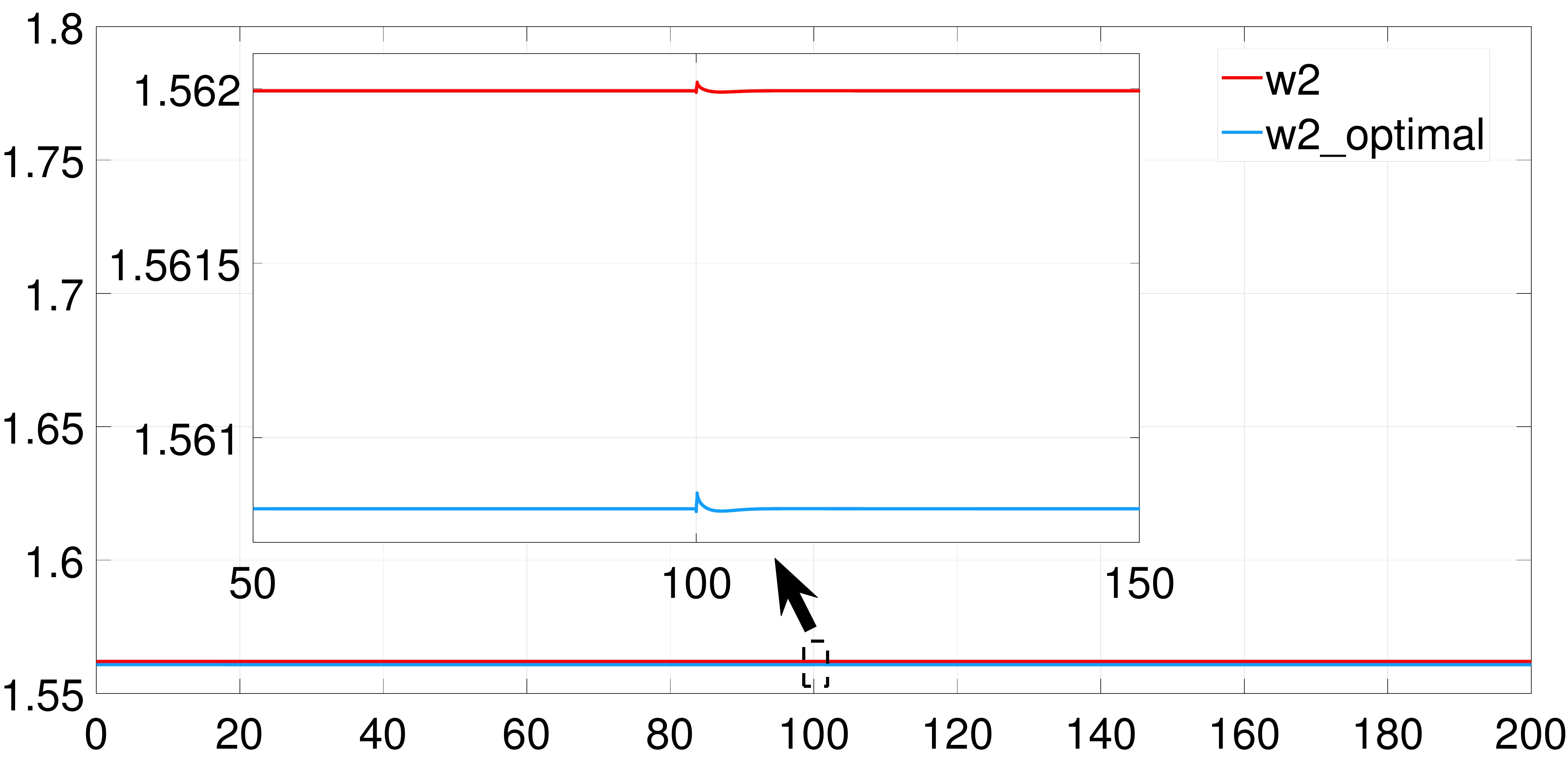}
    \caption{The generator speed (pu) of PMSG2 (Voltage service)}
    \label{fig:Sim_2_PMSG2_w2}
\end{figure}


\begin{figure}[H]
    \centering
    \includegraphics[scale=0.14]{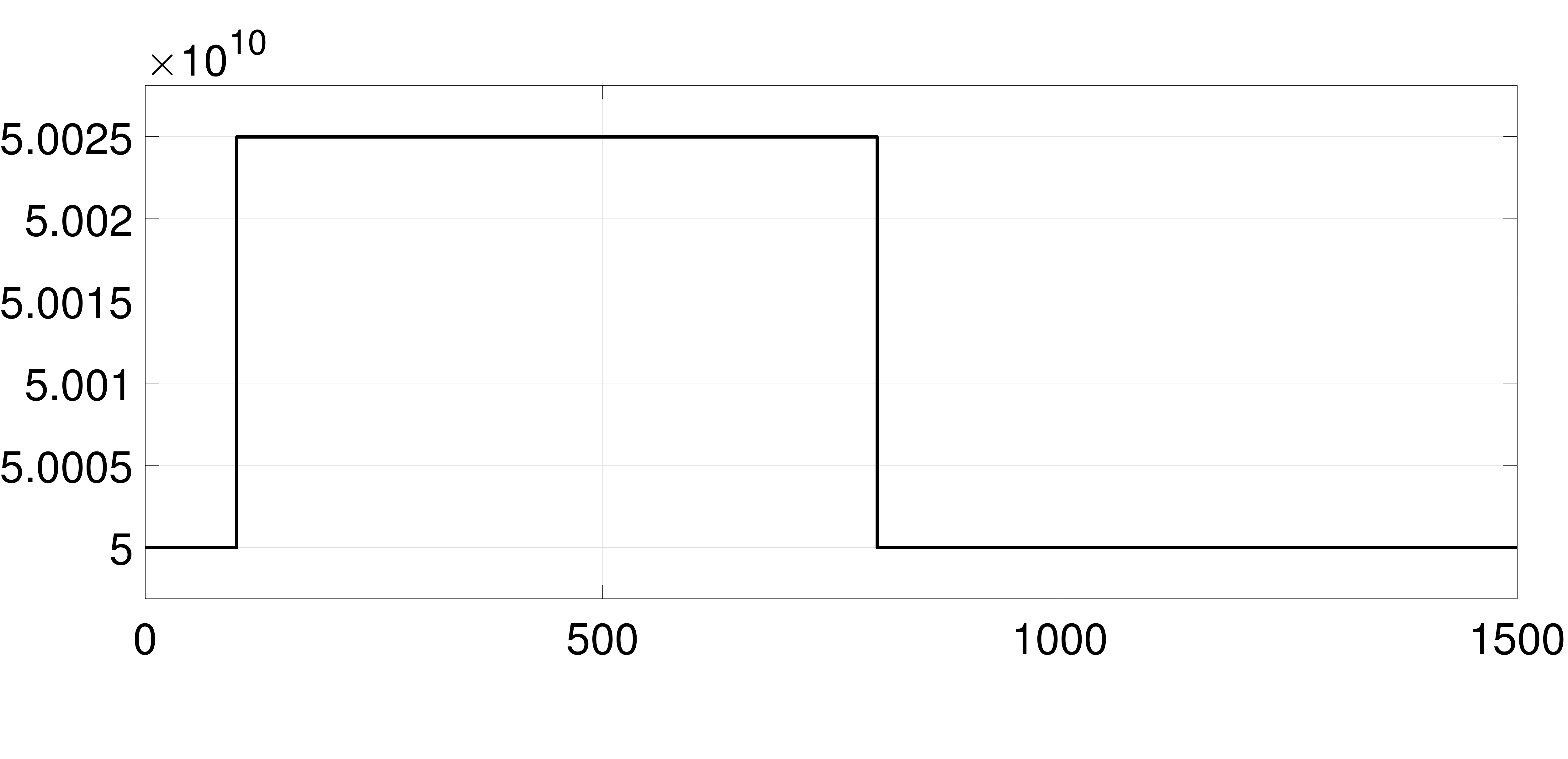}
    \caption{The grid load changing scenario (frequency service)}
    \label{fig:Sim_3_P_L}
\end{figure}

\begin{figure}[H]
    \centering
    \includegraphics[scale=0.14]{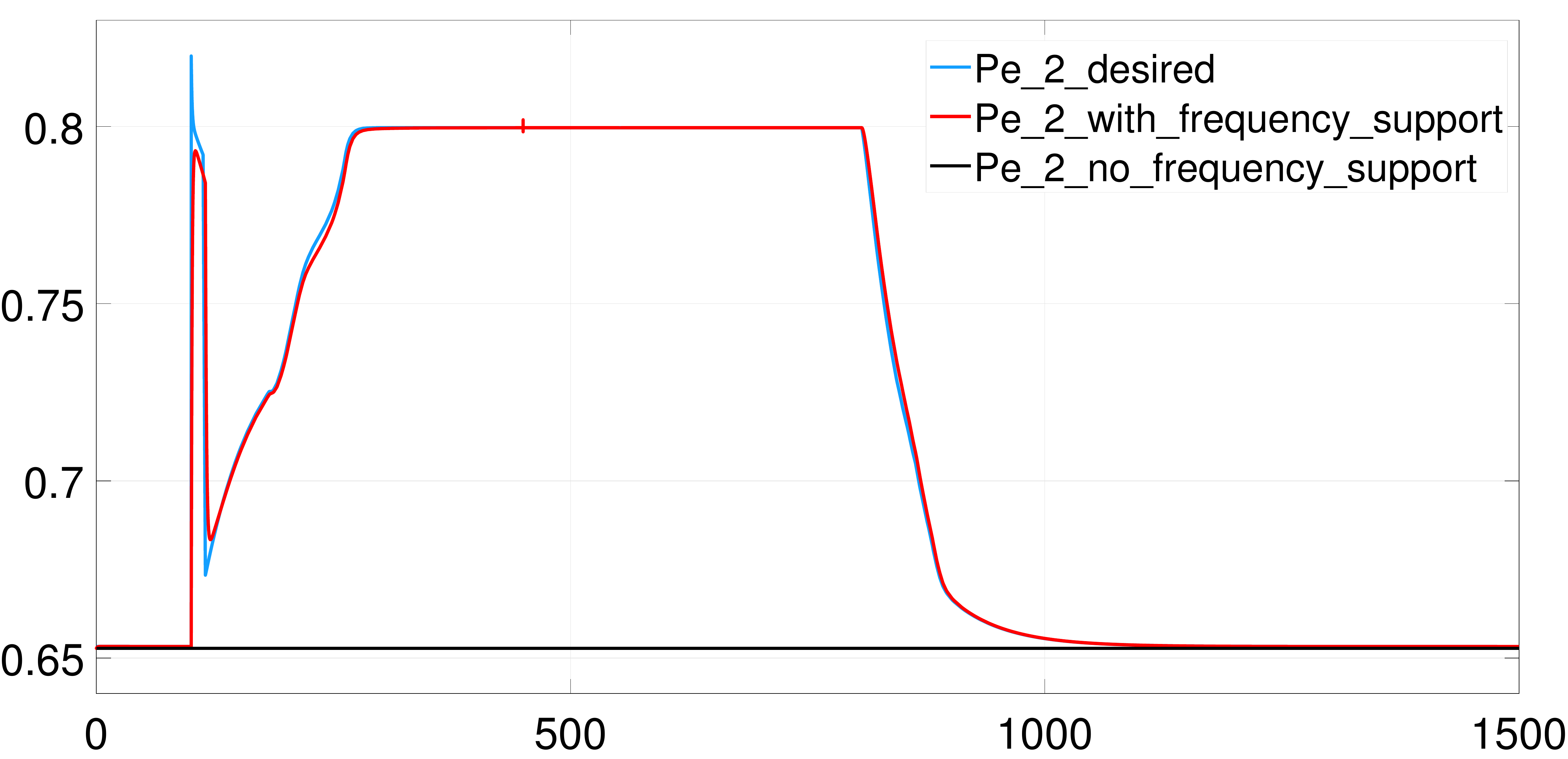}
    \caption{Active power (pu) of PMSG2 (frequency service)}
    \label{fig:Sim_3_PMSG2_P2}
\end{figure}

\begin{figure}[H]
    \centering
    \includegraphics[scale=0.14]{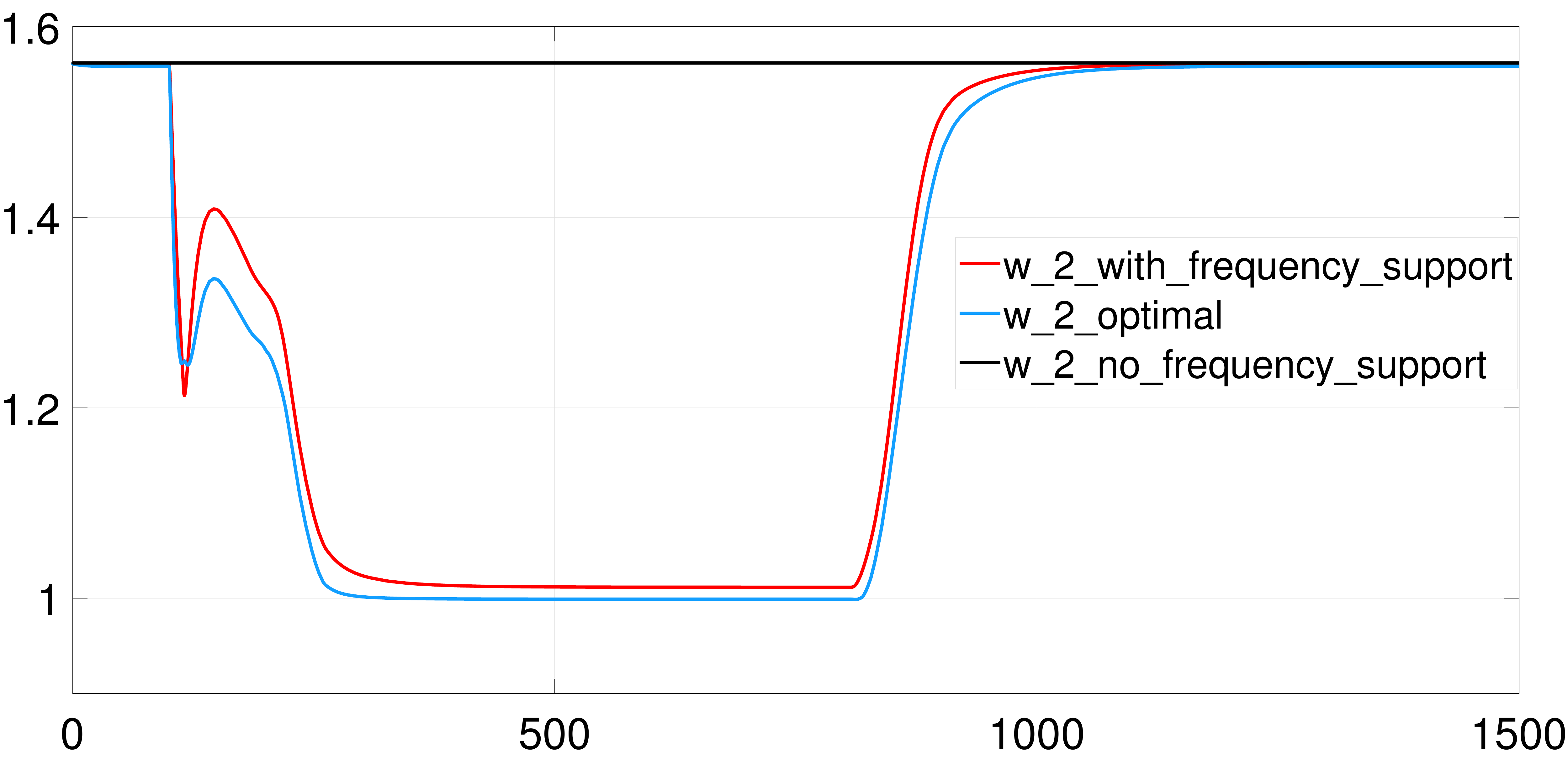}
    \caption{Generator speed of PMSG2 (frequency service)}
    \label{fig:Sim_3_PMSG2_w2}
\end{figure}

\begin{figure}[H]
    \centering
    \includegraphics[scale=0.14]{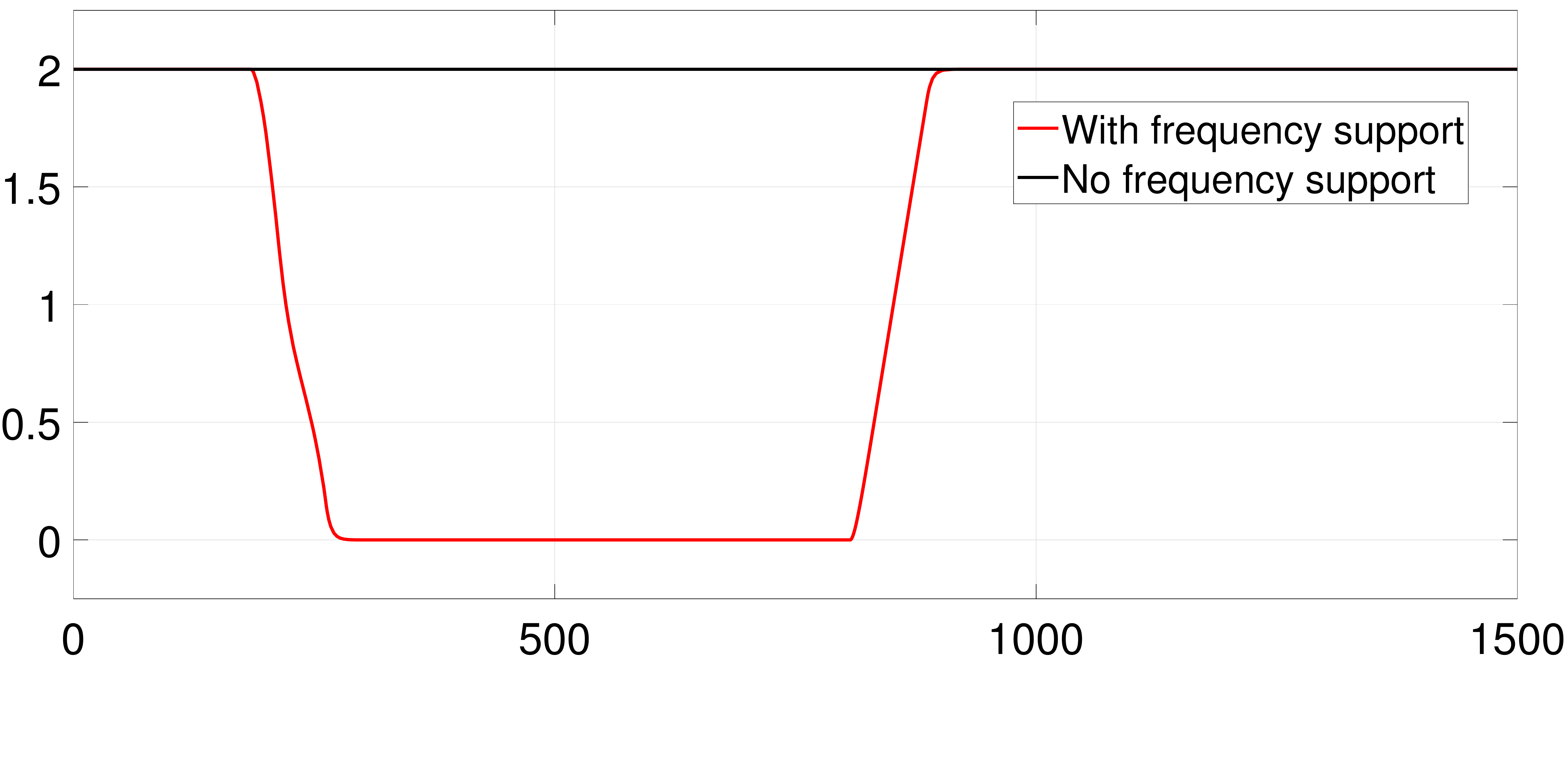}
    \caption{Pitch angle (deg) of PMSG2 (frequency service)}
    \label{fig:Sim_3_PMSG2_pitch_angle}
\end{figure}

\begin{figure}[H]
    \centering
    \includegraphics[scale=0.14]{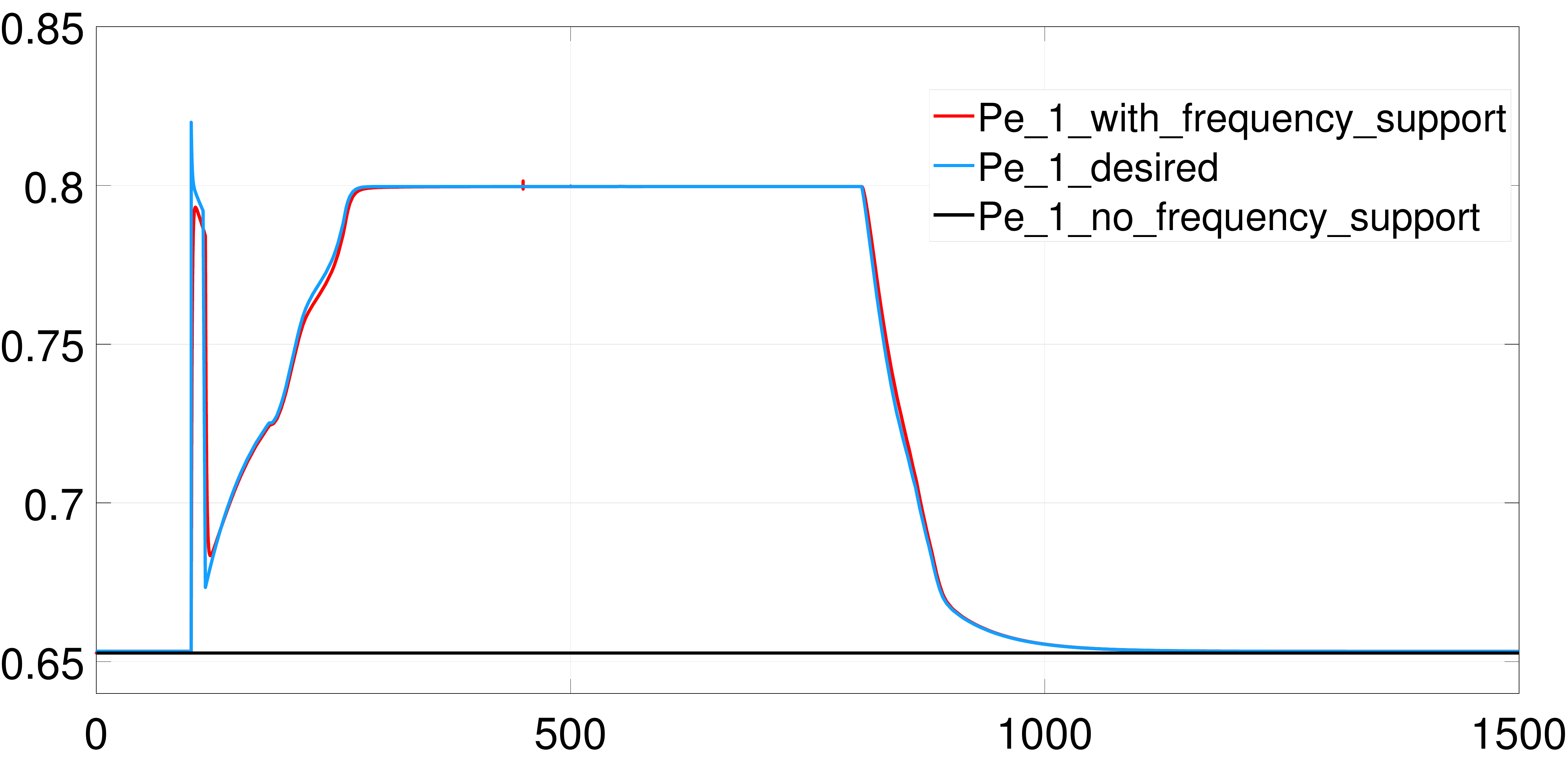}
    \caption{Active power (pu) of PMSG1 (frequency service)}
    \label{fig:Sim_3_PMSG1_P1}
\end{figure}

\begin{figure}[H]
    \centering
    \includegraphics[scale=0.14]{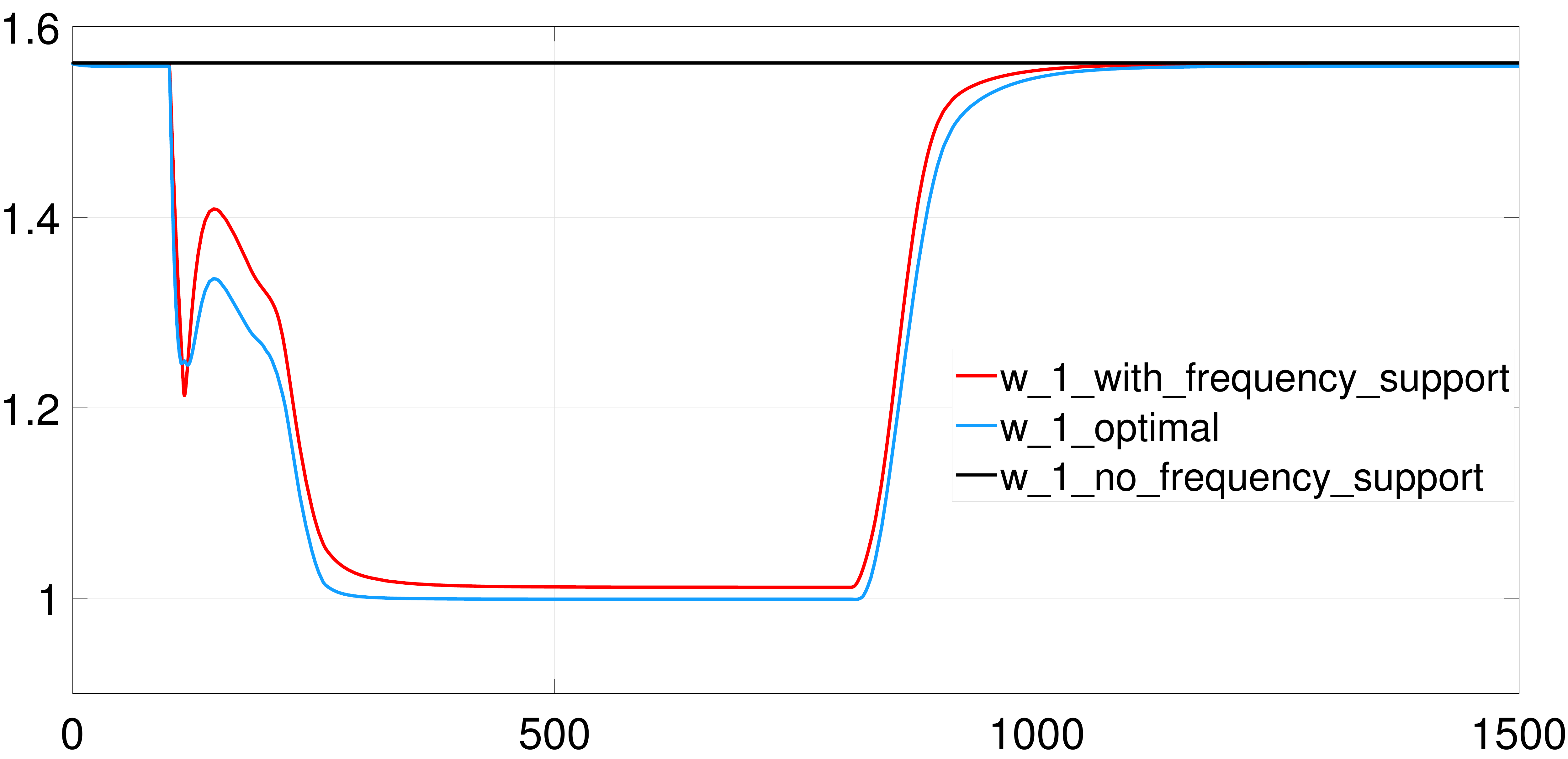}
    \caption{Generator speed of PMSG1 (frequency service)}
    \label{fig:Sim_3_PMSG1_w1}
\end{figure}

\begin{figure}[H]
    \centering
    \includegraphics[scale=0.14]{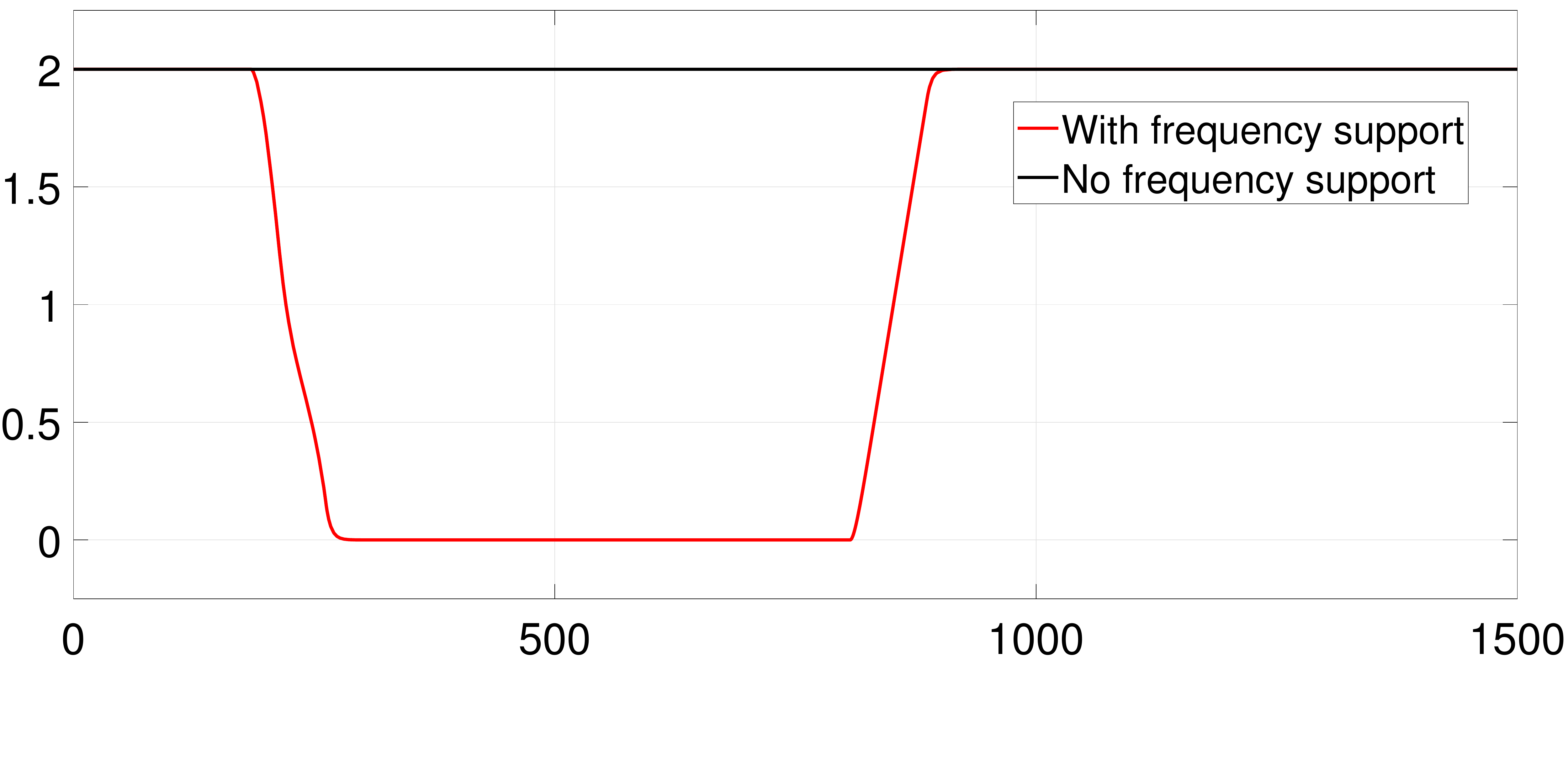}
    \caption{Pitch angle (deg) of PMSG1 (frequency service)}
    \label{fig:Sim_3_PMSG1_pitch_angle}
\end{figure}

\begin{figure}[H]
    \centering
    \includegraphics[scale=0.14]{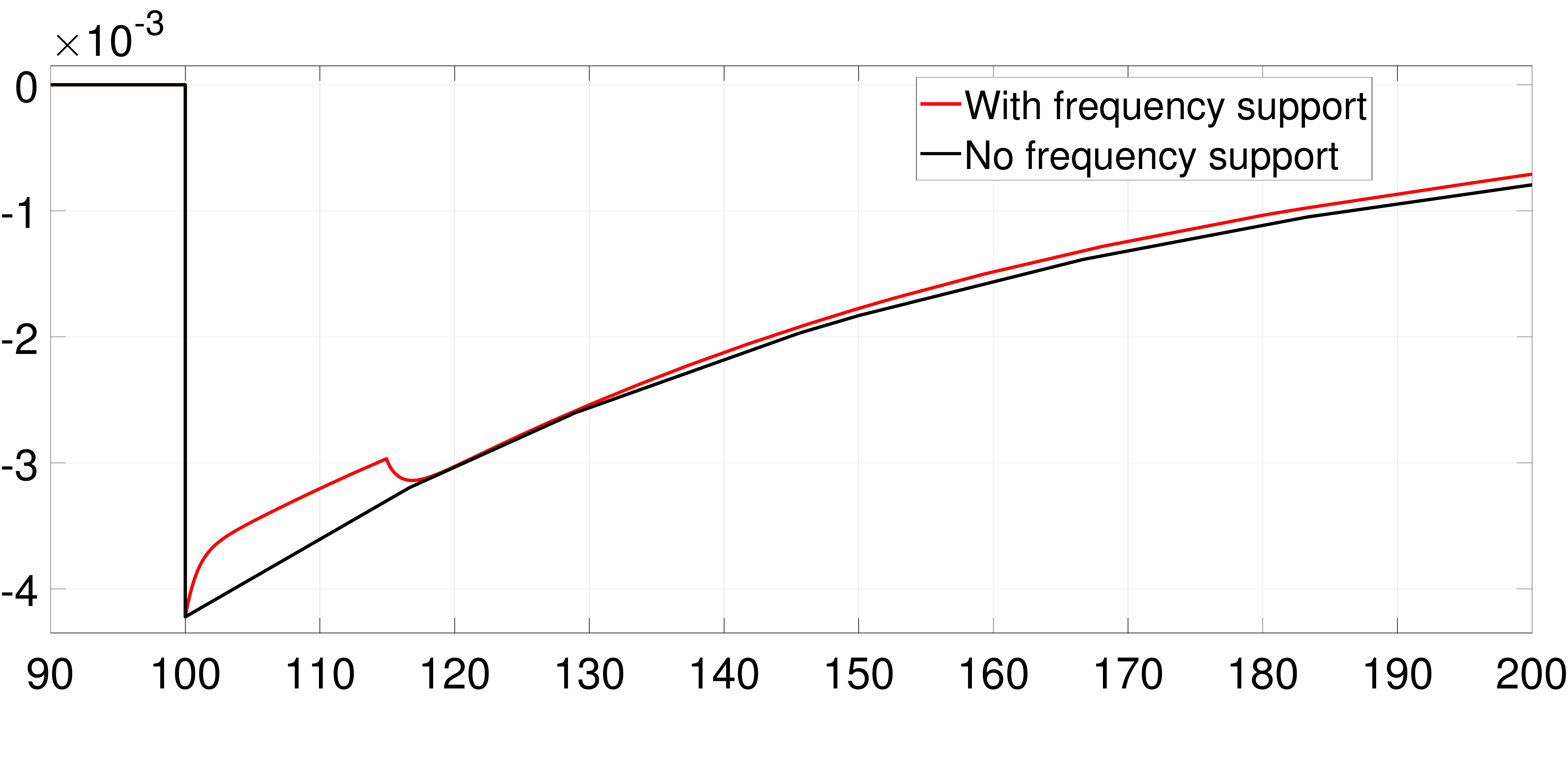}
    \caption{RoCoF (pu) behavior (frequency service)}
    \label{fig:Sim_3_RoCoF}
\end{figure}

\begin{figure}[H]
    \centering
    \includegraphics[scale=0.14]{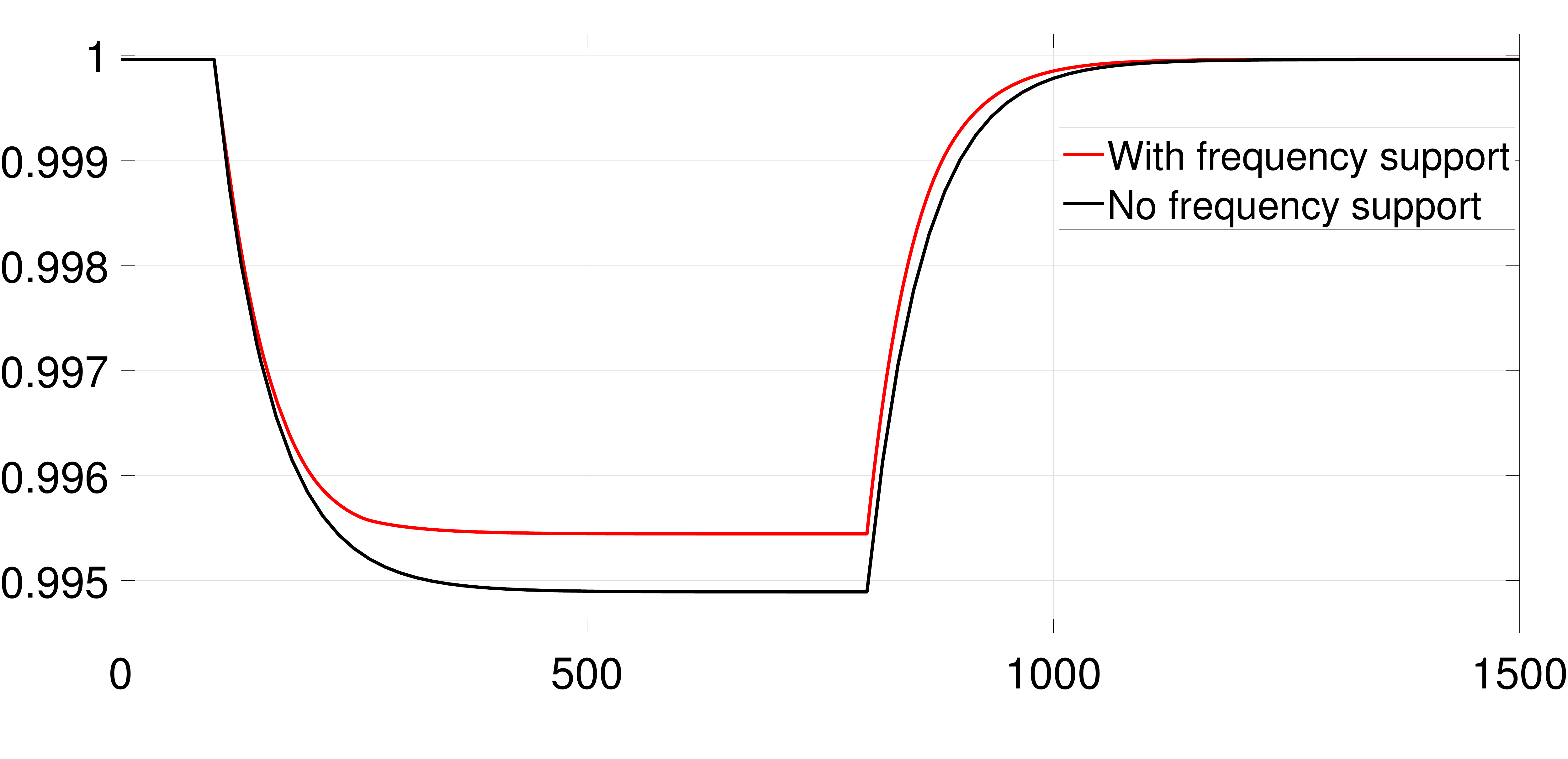}
    \caption{Frequency (pu) of the grid (frequency service)}
    \label{fig:Sim_3_wf}
\end{figure}


\begin{figure}[H]
    \centering
    \includegraphics[scale=0.14]{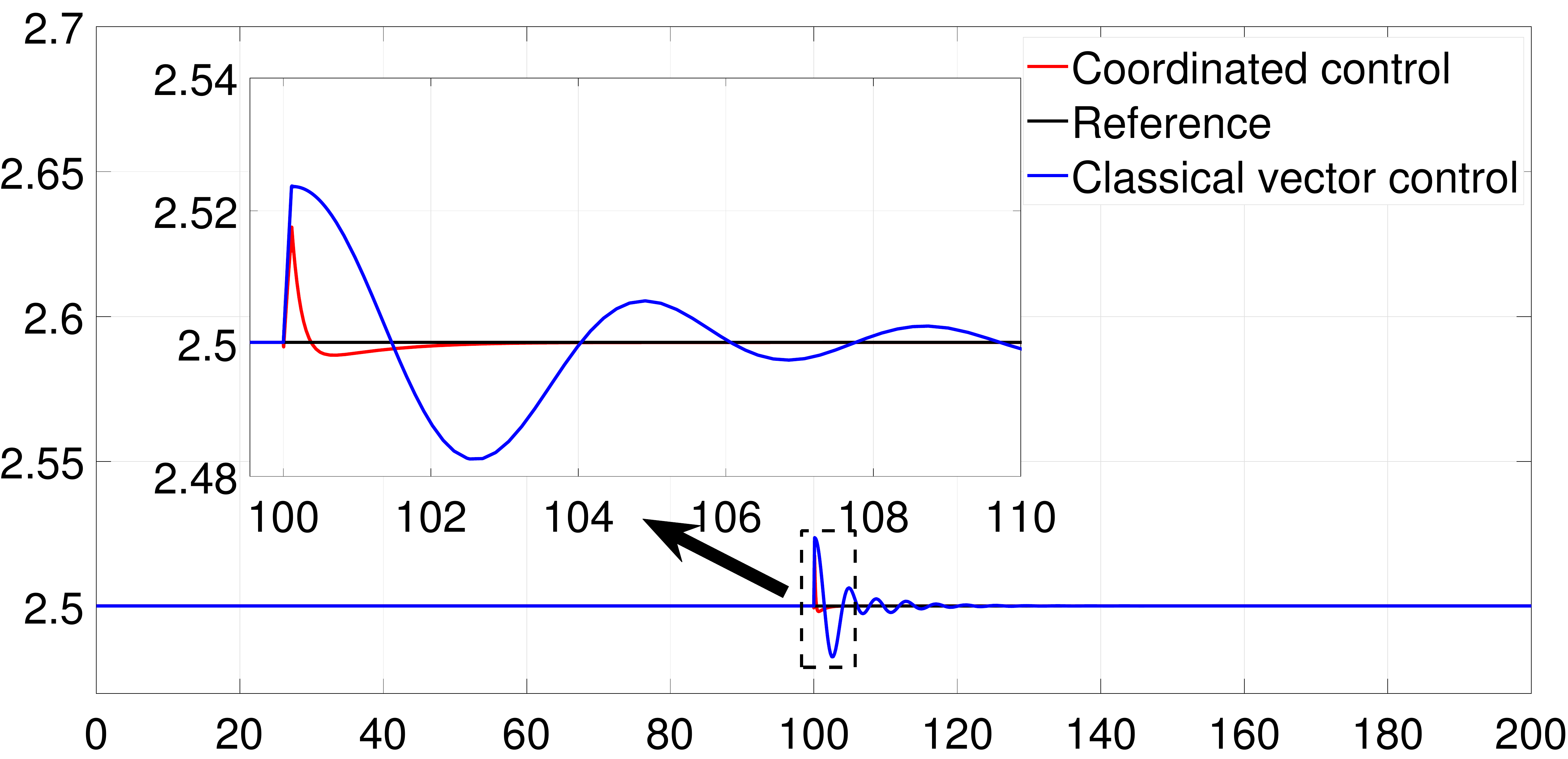}
    \caption{DC voltage (pu) of PMSG1 (coordinated control vs vector control)}
    \label{fig:Sim_4_PMSG1_v_DC1}
\end{figure}

\begin{figure}[H]
    \centering
    \includegraphics[scale=0.14]{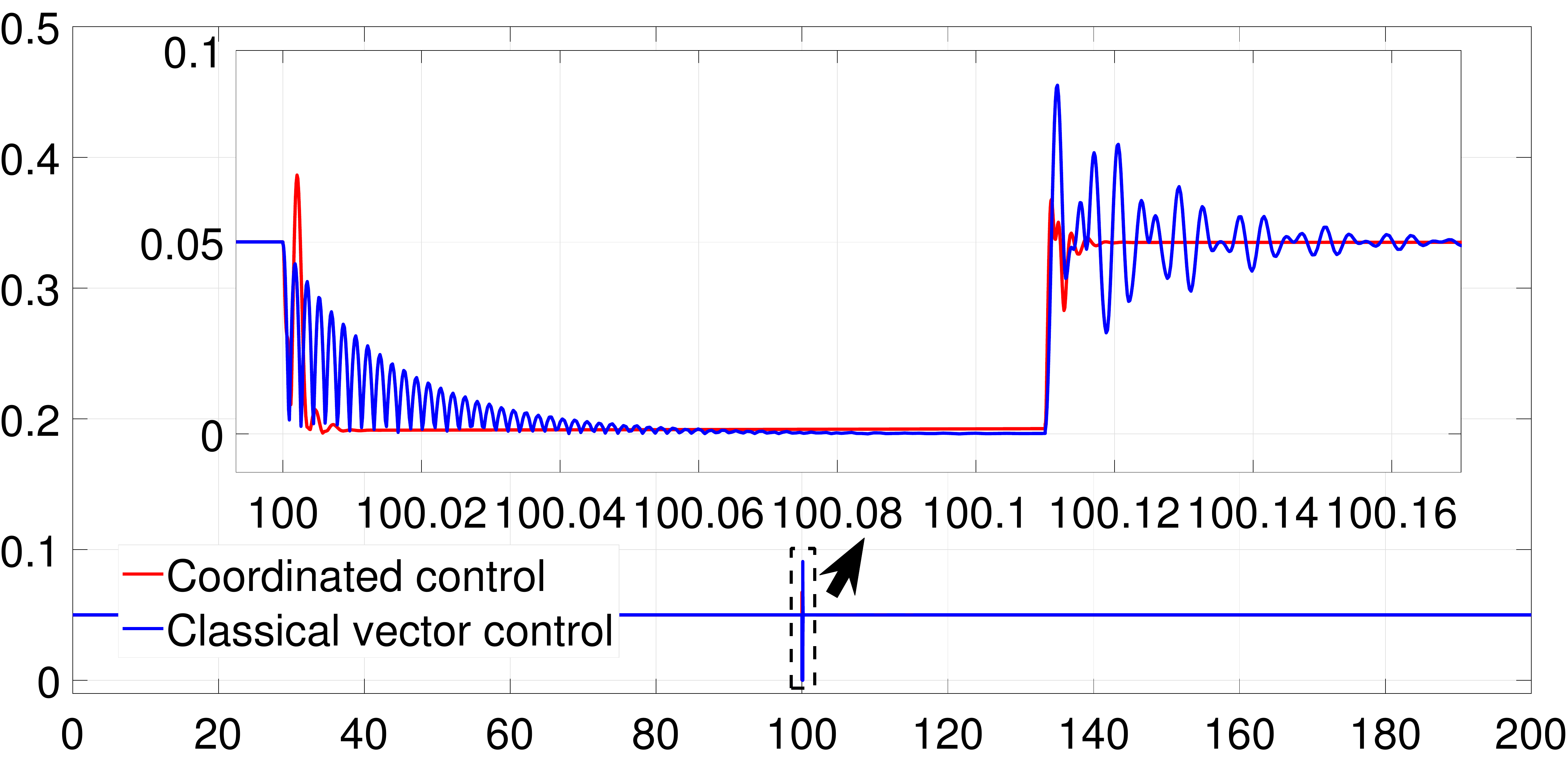}
    \caption{AC terminal voltage $V_{r1}$ (pu) of PMSG1 (coordinated control vs vector control)}
    \label{fig:Sim_4_PMSG1_V_r1}
\end{figure}

\begin{figure}[H]
    \centering
    \includegraphics[scale=0.14]{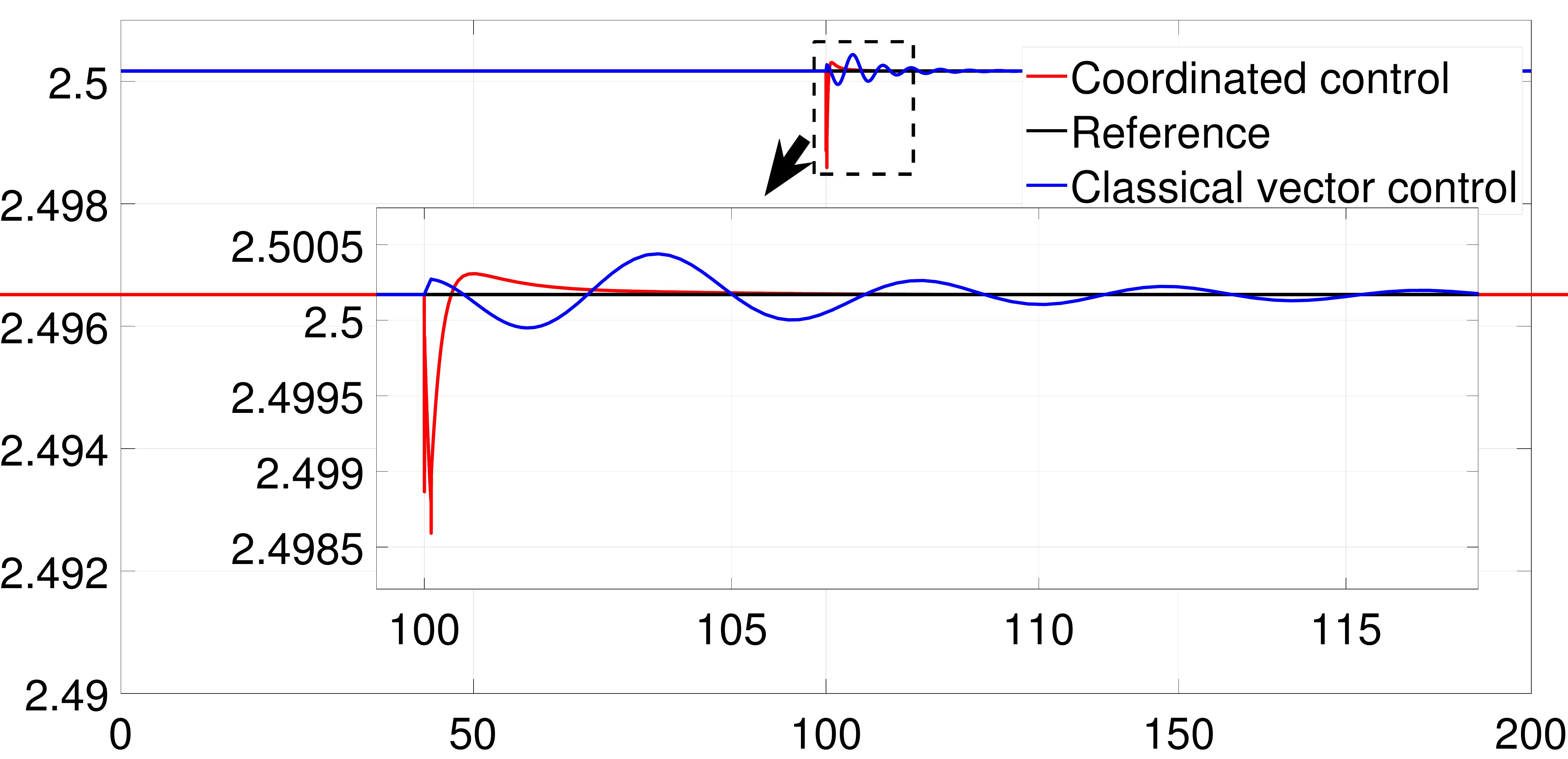}
    \caption{DC voltage (pu) of PMSG2 (coordinated control vs vector control)}
    \label{fig:Sim_4_PMSG2_v_DC2}
\end{figure}

\begin{figure}[H]
    \centering
    \includegraphics[scale=0.14]{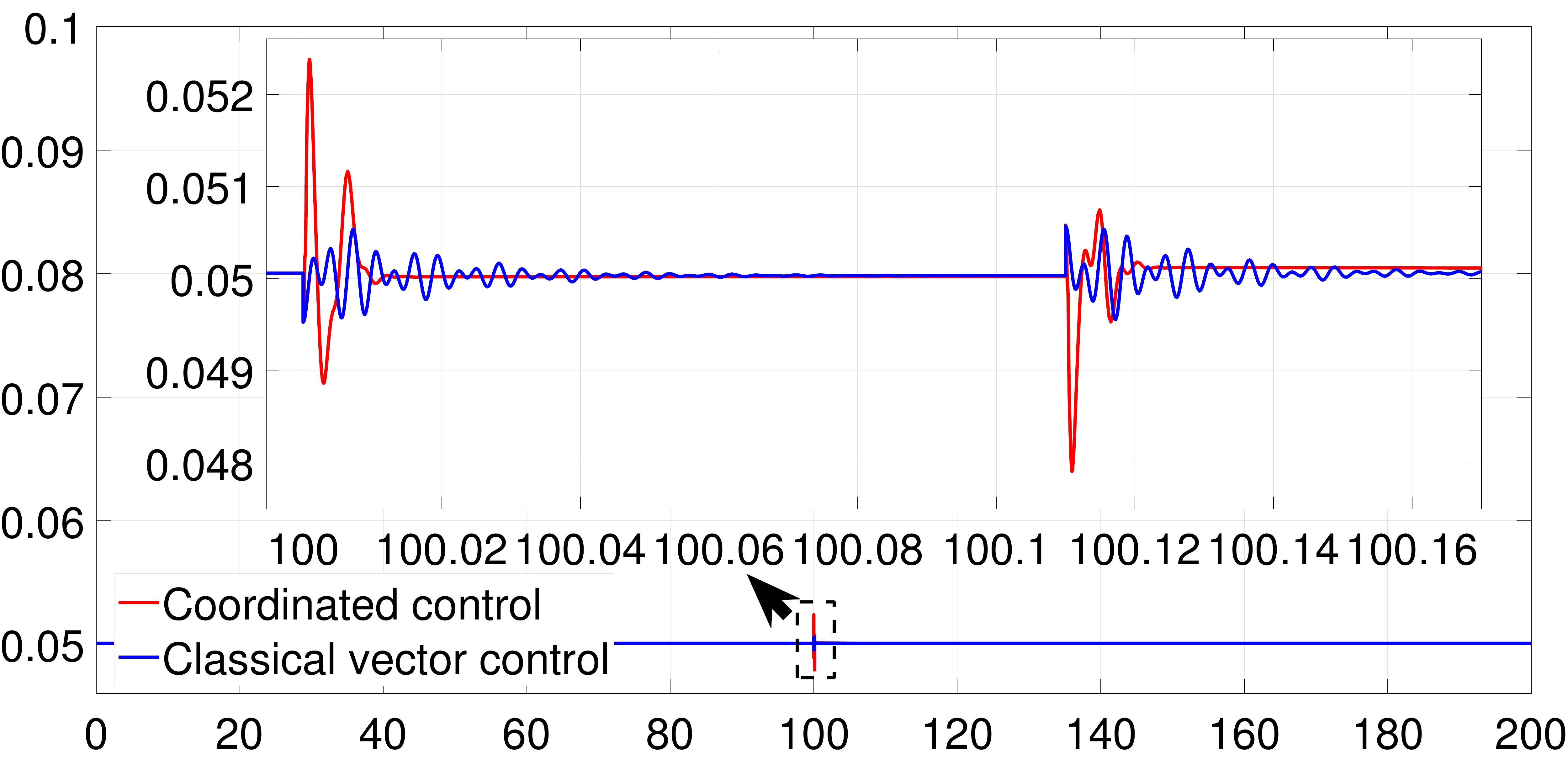}
    \caption{AC terminal voltage $V_{r2}$ (pu) of PMSG2 (coordinated control vs vector control)}
    \label{fig:Sim_4_PMSG2_V_r2}
\end{figure}

\section{Conclusions}
To address full participation of PPMs in grid ancillary services, we have first critically reviewed the classic regulation framework and main differences at dynamic level between PPMs and classic synchronous generators. It has been concluded that old hypothesis of separation between V and f dynamics should be skipped at all levels: modeling and synthesis of controls. Coordination should be achieved at these levels as well as for other aspects. As a consequence, a \textit{new control model} which gathers all dynamics of interest has been proposed. Based on it, a strategy of robust advanced model-based control was proposed. The proposed approach is based on the overall \textit{time-scales separation} of the dynamics of the global system and of its control objectives. This makes a significant difference with the classic vector control and related approaches which propose several simple PI loops around each actuator of the global system. Maximum coordination is achieved as all actuators and dynamics of interest are taken into account in the control model on which the synthesis of the control is based. This coordination is among the 2 converters of a same PPM or even among converters (and other dynamics) of several PPMs when control is addressed for a class of PPMs as in the example treated.

This aspect will be fully developped in following work and publications in the framework of H2020 POSYTYF project where a class of PPMs will be controlled in a new concept called Dynamic Virtual Power Plant to maximize RES participation to grid ancillary services.

The proposed control methodology and concepts are general and can be used for any type of PPM or RES connected to the grid with power converters. Elements of the time-decoupling diagram in Fig. \ref{fig:Time_decoupling} should be specifically used for each case. For example, in case of a storage element (battery), it is expected to control only the very fast dynamics. 

Synthesis of the control in close relation to the time-scale of phenomena and objectives paves the way to participation of PPMs and RESs to grid ancillary services. This means to really integrate the new sources into today existing secondary V and f controls, at the same level of exigences as the classic synchronous generators and not only to provide some implicit grid support, which is difficult to quantify and relay on it. Integration of PPMs into the secondary level control, i.e., stage "`slow"' in Fig. \ref{fig:Time_and_space_separation} is also a next step in our work in POSYTYF project.

The presented methodology is a \textit{model-based} one which is open to any advanced control. Here loop-shaping was achieved by H-infinity synthesis. In next work we will add fuzzyfication to improve behavior with respect to nonlinearities and variations of system operating conditions. Also, extended validation in hardware-in-the loop environment will be next done.

\appendices

\section{Classic control of wind PPMs}

\subsection{Droop control with deloading methods}\label{DroopControl}

Droop control, since the era of classical synchronous generators, is an important control loop to regulate frequency based on active power control of a grid \cite{Kundur}. However, with RES, especially wind generators, to minimize the overall cost, the main goal was to maximize power generation using MPPT \cite{2}-\cite{3}, \cite{7}, and this leaves no room for active power regulation in case of grid incident. Hence, a new method emerged, which try to reserve an amount of energy by regulating the produced energy at lower level than MPPT, called \textit{deloading control} \cite{7}, \cite{18}, \cite{35}-\cite{38}.

10\% is a typical value for the reserve of the deloading control using generator speed control. Higher reserve can be ensured if pitch angle control is used in addition. 

Firstly, consider the pitch angle is kept at 0 (deg), the deloading relationship between wind speed, generator speed and wind turbine mechanical power is given in Fig. \ref{fig:2D_Deloading_control}.

\begin{figure}[H]
    \centering
    \includegraphics[scale=0.18]{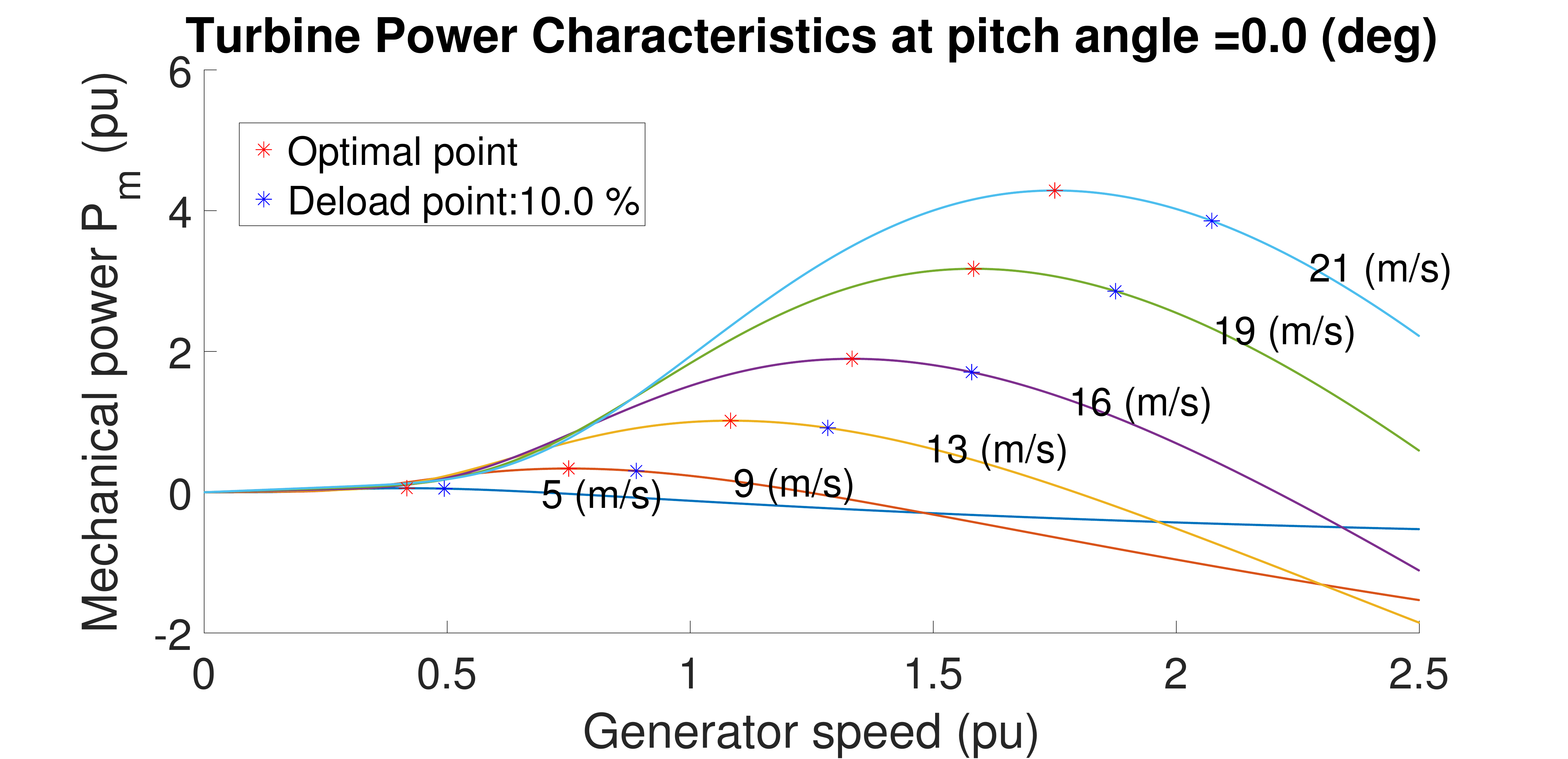}
    \caption{The output power at different wind speeds at MPPT and at 10\% deload}
    \label{fig:2D_Deloading_control}
\end{figure}

To understand the deloading strategy, consider the wind turbine is operating at wind speed of 16m/s. The wind turbine will produce power at 10$\%$ deloading point A of Fig. \ref{fig:Generator_speed_deloading} to reserve 10$\%$ of its power capacity.  

\begin{figure}[H]
    \centering
    \includegraphics[scale=0.18]{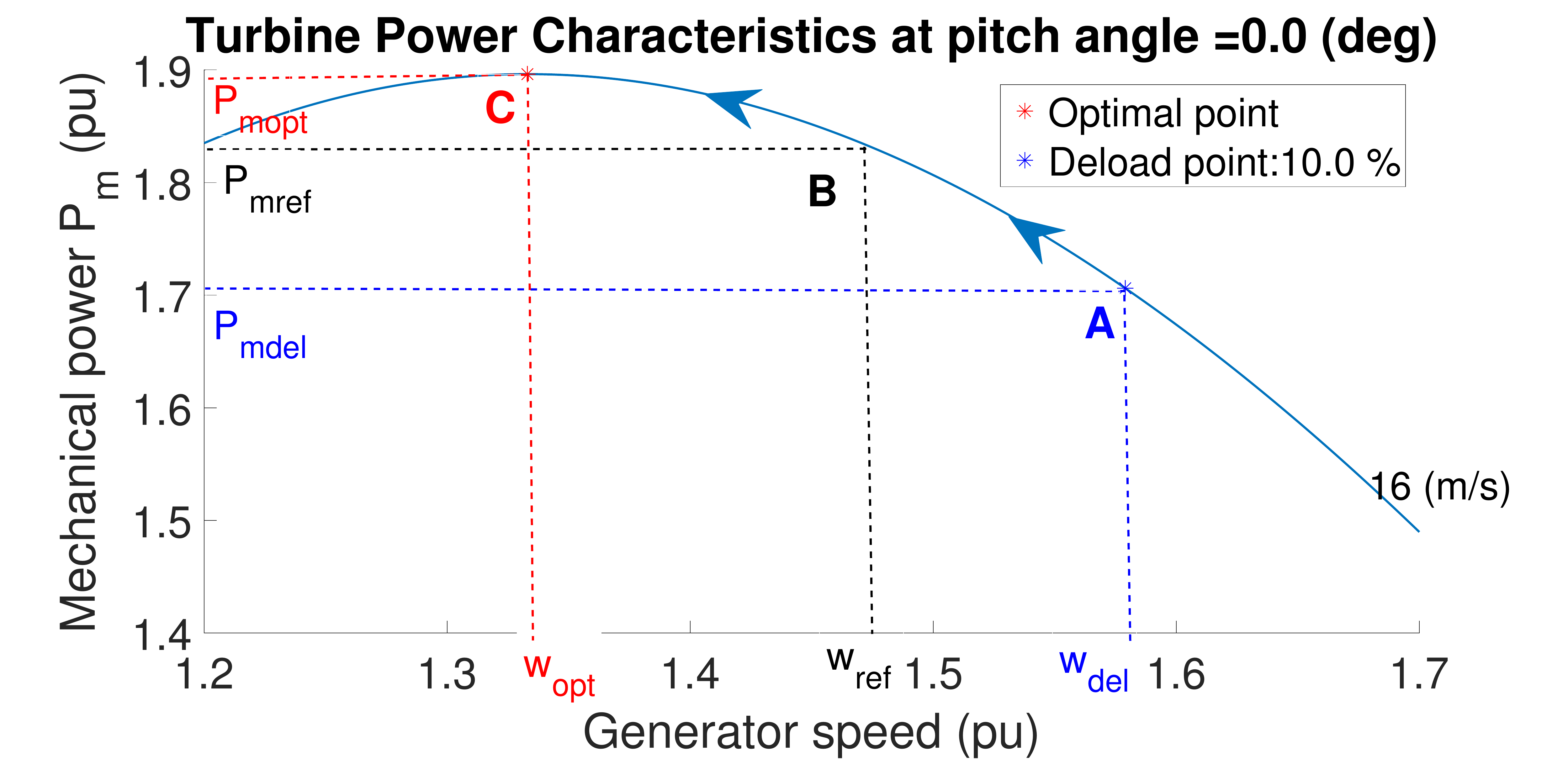}
    \caption{The generator speed deloading strategy at wind speed of 16 m/s}
    \label{fig:Generator_speed_deloading}
\end{figure}

To utilize this reserved power for frequency support, a droop relationship is established \cite{Kundur}

\begin{equation}
\Delta {P_{ref}} =  - \frac{1}{R}\Delta f,
\label{eq:droop_relationship}
\end{equation}

where $\Delta {P_{ref}}$ is the required amount of change in power when there is a change in frequency $\Delta f$.

A consequence of a shortage in power of the the grid is a fall in grid frequency \cite{Kundur}. This fall will partly be compensated by the reserved power of wind turbine. At this moment, through the droop relationship \eqref{eq:droop_relationship}, the wind turbine will have to move its operating point from A to B to increase its mechanical power. The operating point is now at

\begin{equation}
P_{mref} =  P_{mdel} + \Delta {P_{ref}} = P_{mdel} - \frac{1}{R}\Delta f.
\label{eq:P_mref_deload_relationship}
\end{equation}

This $P_{mref}$ will be used to determine desired generator speed $\Omega_{ref}$, and through generator speed control to shift the operating point of wind turbine. This action can easily be achieved by a look-up table \cite{7} \cite{35} - \cite{38}, but with a cost in delay due to the time for searching the equivalent value of $\Omega_{ref}$. To overcome this, because the two points of optimal and deloading points are very close, it is possible to consider that the three points A, B and C are collinear. As a consequence, the relationship between $P_{mref}$ and $\Omega_{ref}$ is \cite{7}

\begin{equation}
{P_{mref}} = {P_{mdel}} + \left( {{P_{mopt}} - {P_{mdel}}} \right)\frac{{{\Omega_{del}} - {\Omega_{ref}}}}{{{\Omega_{del}} - {\Omega_{opt}}}}
\label{eq:P_mref_deload_generator_speed}
\end{equation}

Based on \eqref{eq:P_mref_deload_generator_speed}, the value of $\Omega_{ref}$ is

\begin{equation}
{\Omega_{ref}} = {\Omega_{del}} - \left( {{\Omega_{del}} - {\Omega_{opt}}} \right)\left( { - \frac{1}{R}\frac{{\Delta f}}{{{P_{mopt}} - {P_{mdel}}}}} \right)
\label{eq:w_ref_deload_relationship}
\end{equation}

Secondly, consider all the deloading relationships between wind speed, generator speed, pitch angle and wind turbine mechanical power in the 3D scheme of Fig. \ref{fig:3D_Deloading_control}.

\begin{figure}[H]
    \centering
    \includegraphics[scale=0.11]{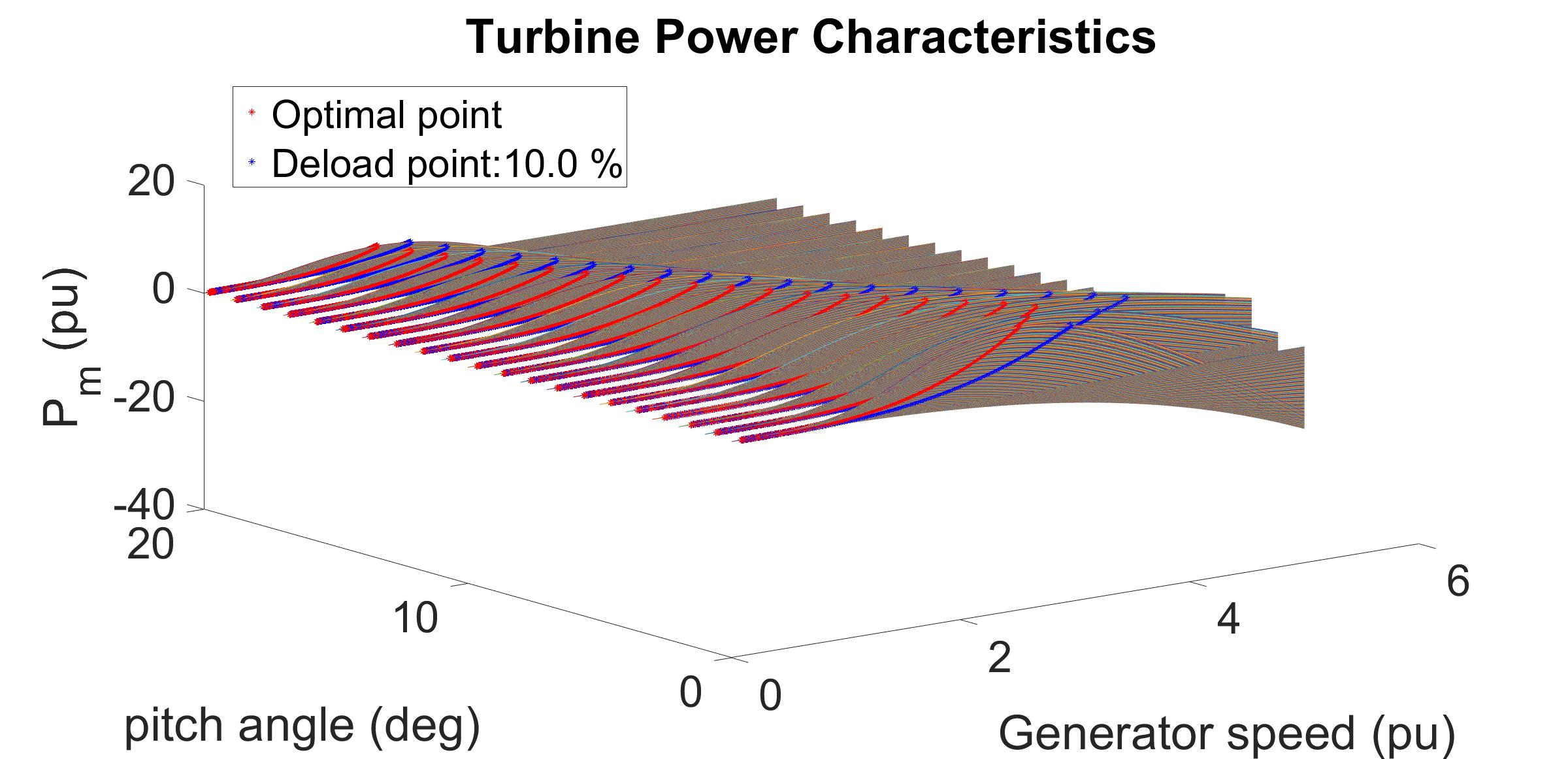}
    \caption{The output power at different wind speeds (1-30 m/s) at MPPT and at 10\% deload}
    \label{fig:3D_Deloading_control}
\end{figure}

Again, to deload the wind turbine using pitch angle control, the pitch angle will be kept at deloading point instead of optimal point (0 deg) \cite{7}, \cite{35} - \cite{38}. To understand the deloading pitch angle strategy, consider the wind turbine is operating at constant wind speed 12 (m/s) for Fig. \ref{fig:Pitch_angle_deloading}.

\begin{figure}[H]
    \centering
    \includegraphics[scale=0.18]{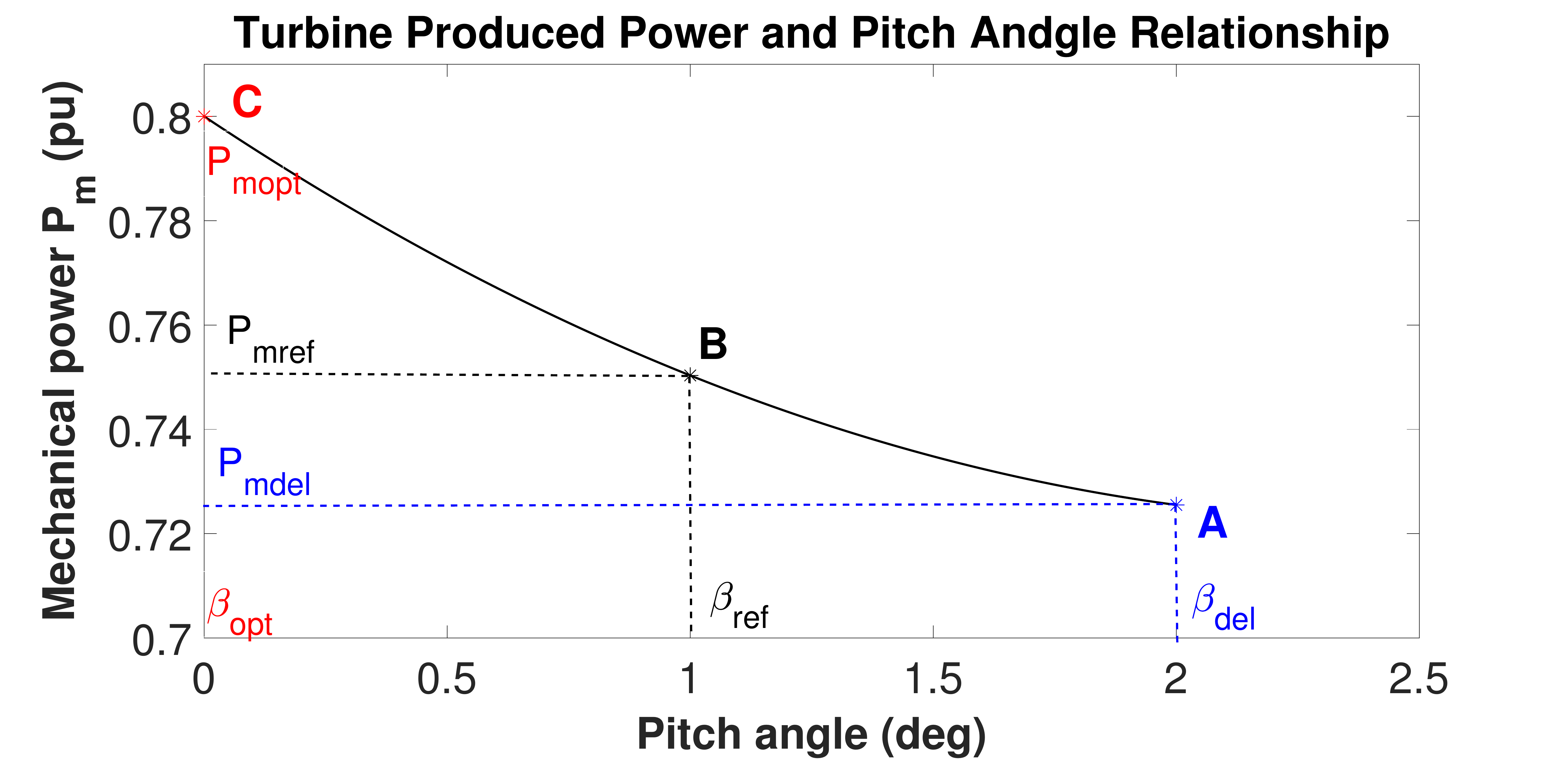}
    \caption{The pitch angle deloading strategy at wind speed of 12 m/s}
    \label{fig:Pitch_angle_deloading}
\end{figure}

At steady-state, wind turbine will operate at point A  with pitch angle $\beta_{del}$ and produce power at $P_{del}$. Following a request to increase the power \eqref{eq:P_mref_deload_relationship} via the droop control \eqref{eq:droop_relationship}, the pitch angle will be changed from point A to point B.

Once again, a look-up table can be used, but this results in introducing significant delay into the control loop.  This problem can be dealt with the same method as in \eqref{eq:P_mref_deload_generator_speed}. The relationship between desired produced power $P_{mref}$ and desired pitch angle $\beta_{ref}$ is then

\begin{equation}
{P_{mref}} = {P_{mdel}} + \left( {{P_{mopt}} - {P_{mdel}}} \right)\frac{{{\beta _{del}} - {\beta _{ref}}}}{{{\beta _{del}} - {\beta _{opt}}}}
\label{eq:P_mref_deload_pitch_angle}
\end{equation}

From droop control relationship \eqref{eq:P_mref_deload_relationship} and \eqref{eq:P_mref_deload_pitch_angle}, the desired pitch angle $\beta_{ref}$ is

\begin{equation}
{\beta _{ref}} = {\beta _{del}} - \left( {{\beta _{del}} - {\beta _{opt}}} \right)\left( { - \frac{1}{R}\frac{{\Delta f}}{{{P_{mopt}} - {P_{mdel}}}}} \right).
\label{eq:beta_ref_deload_relationship}
\end{equation}

\subsection{Inertia control for wind turbine}\label{InertiaControl}

Currently, there are two popular types of fast support for RoCoF improvement for RES: hidden inertia control \cite{7} \cite{8}-\cite{12} and fast power reserve \cite{7}, \cite{13}-\cite{16}. Because of the lack of the space, only the first one is presented and used (Fig. \ref{fig:Hidden_inertia}).

Hidden inertia \cite{9}, or virtual inertia \cite{8}, or inertia emulation \cite{7} is basically a way to react to the change of RoCoF, by feedback RoCoF to the reference of electrical torque or electrical active power

\begin{equation}\label{eq:Hidden_inertia}
\Delta {P_{eref}} = K\frac{{df}}{{dt}},
\end{equation}

where $\Delta {P_{eref}}$ is the "needed" power to compensate the drop in RoCoF, $f$ is the nearby grid frequency and $K$ is the inertia response loop gain, which depends on inertia constant of wind turbine and may depend on frequency $f$ as in \cite{7} - \cite{9}. The hidden inertia control can also be modified and utilized at the same time with droop control to achieve better frequency support \cite{7}-\cite{8}.

\begin{figure}[H]
    \centering
    \includegraphics[scale=0.5]{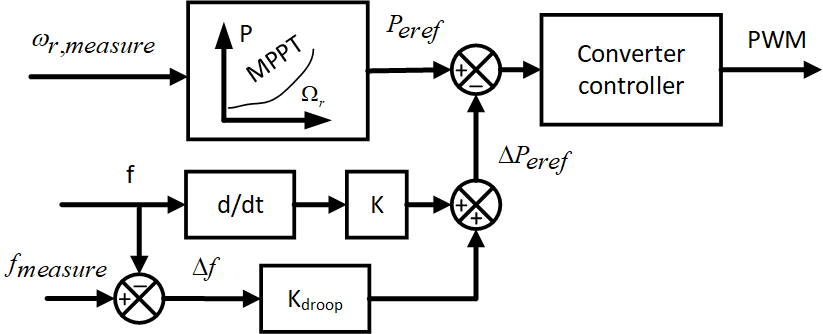}
    \caption{The hidden inertia control structure \cite{7}}
    \label{fig:Hidden_inertia}
\end{figure}

The gain $K$ is estimated starting from the change of power of the grid due to frequency variation

\begin{equation}
\Delta {P_L} = 2H\frac{{df}}{{dt}},
\end{equation}

 where $H$ is the inertia constant of the grid, $f$ is grid frequency and $\Delta {P_L}$ is the total power change. The idea behind inertia control is to immediately produce the same amount of total power to keep the RoCoF as small as possible. This means that the required produced power should be
 
\begin{equation}
\Delta {P_{eref}} =  - \Delta {P_L} =  - 2H\frac{{df}}{{dt}}
\label{eq:P_eref_Hidden_inertia}
\end{equation}

Hence, the gain $K$ and droop gain are

\begin{equation} \label{eq:RoCoF_and_droop_gain}
\left\{ \begin{array}{l}
K =  - 2H\\
{K_{droop}} =  - \frac{1}{R}
\end{array} \right.
\end{equation}

\subsection{Vector control for converters}\label{subsectionVectorControl}

Vector control is a very popular method in the field of power converters control \cite{vector_control_1}, \cite{vector_control_3} - \cite{vector_control_6}. It is structured in two control loops decoupled in time by the control itself: one fast for the current, called \textit{inner loop} and one slower for the voltage (and/or reactive and active power), called \textit{outer loop}. They are given in Fig. \ref{fig:vector_control_structure} for the grid-side converter.

\begin{figure}[H]
    \centering
    \includegraphics[scale=0.35]{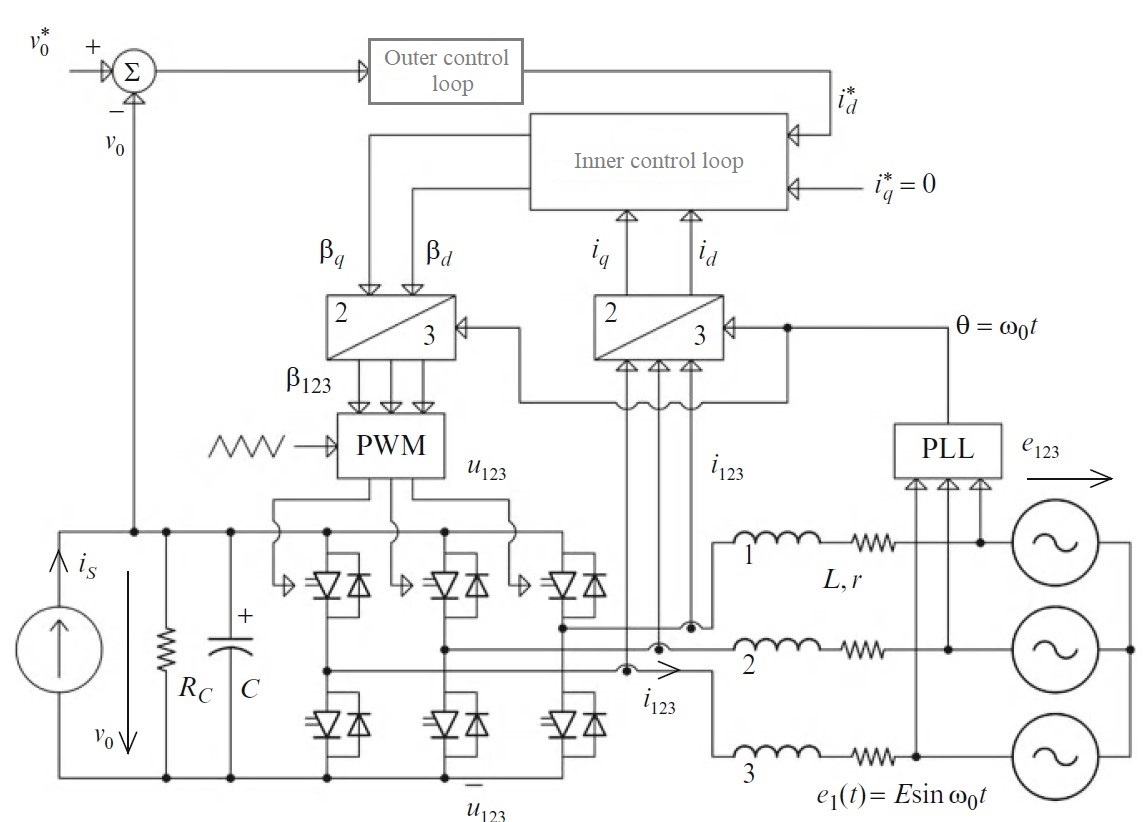}
    \caption{The vector control structure for grid-side converter \cite{vector_control_6}}
    \label{fig:vector_control_structure}
\end{figure}

In the inner loop detailed in Fig. \ref{fig:vecter_control_current_loop}, the currents in each d and q axis are controlled by a PI controller. Two cross connections among these controls are included to diminish interaction. The PI controllers parameters $K_{pC}$ and $T_{iC}$ are chosen by forcing the damping ratio and the time constant of the current closed-loop into desired values \cite{vector_control_5} \cite{vector_control_6}. The references $i^{\star}_d$ and $i^{\star}_q$ are filtered to compensate for the zeros induced in the closed-loop \cite{vector_control_5}, \cite{vector_control_6}.

\begin{figure}[H]
    \centering
    \includegraphics[scale=0.35]{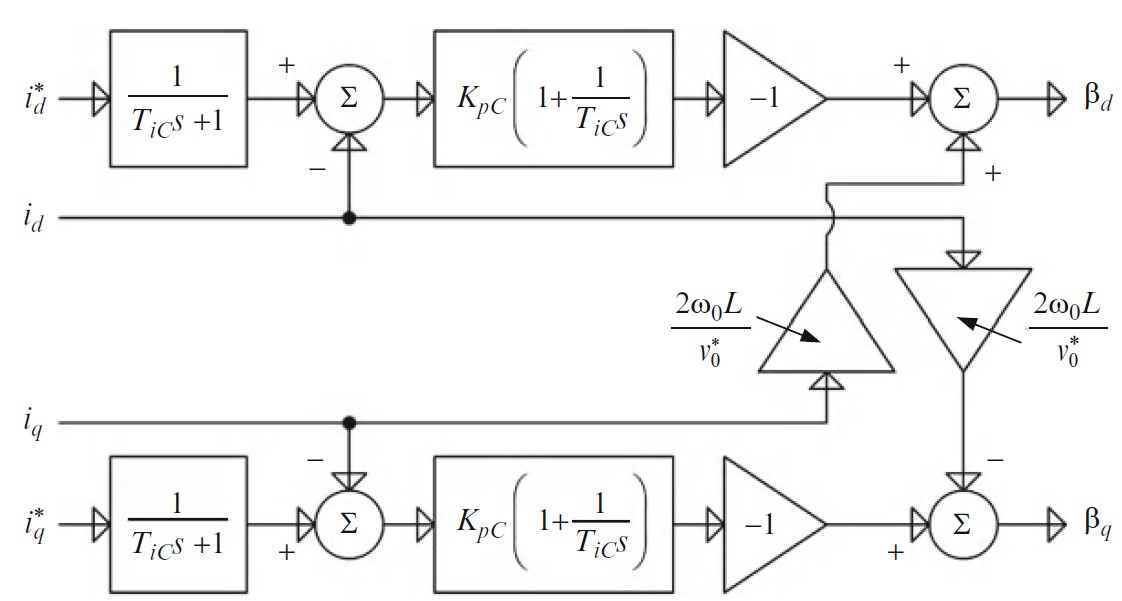}
    \caption{The vector control inner loop of current control \cite{vector_control_6}}
    \label{fig:vecter_control_current_loop}
\end{figure}

The outer loop of voltage control, is actually a loop to produce the desired value of $i_d$, which is the reference for the inner loop in Fig. \ref{fig:vector_control_structure}. As for the inner loop, the PI controllers parameters are chosen by forcing the damping ratio and the time constant of the voltage closed-loop into desired values  \cite{vector_control_5}, \cite{vector_control_6}. The time constant of this outer loop should be around 10-100 times higher that the one of the inner loop to allow independent design of the controller.

%
%
%
%
%

\ifCLASSOPTIONcaptionsoff
  \newpage
\fi

\end{document}